\documentclass[journal]{IEEEtran}
\usepackage[dvips]{graphicx}
\usepackage{epsfig,multirow}
\usepackage{amsmath,amssymb,bm}
\usepackage{amsfonts,dsfont,color,bbm}
\usepackage{algorithm}
\usepackage[noend]{algpseudocode}
\usepackage[font=footnotesize,labelfont=bf,justification=justified,singlelinecheck=false]{caption}
\usepackage{subcaption,slashbox}
\usepackage{cite,epstopdf,epsf,epsfig}
\usepackage{url}
\usepackage{cleveref}

\crefname{equation}{}{}
\algdef{SE}[DOWHILE]{Do}{doWhile}{\algorithmicdo}[1]{\algorithmicwhile\ #1}%
\usepackage{amsthm}

\newtheorem{theorem}{Theorem}

\newtheorem{remark}{Remark}
\newtheorem{proposition}{Proposition}

\renewcommand{\vec}[1]{\mbox{\boldmath${#1}$}}
\newcommand{\norm}[1]{\left|\left|#1\right|\right|}

\usepackage{setspace}

\newcommand{\be}{\begin{equation}}
\newcommand{\ee}{\end{equation}}
\newcommand{\bea}{\begin{eqnarray}}
\newcommand{\eea}{\end{eqnarray}}

\newcommand{\ei}{\end{itemize}}
\newcommand{\bi}{\begin{itemize}}

\newcommand{\sgn}{\text{sgn}}

\let\oldReturn\Return
\renewcommand{\Return}{\State\oldReturn}

\AtBeginDocument{
\addtolength{\abovedisplayskip}{-0.5ex}
\addtolength{\abovedisplayshortskip}{-0.4ex}
\addtolength{\belowdisplayskip}{-0.5ex}
\addtolength{\belowdisplayshortskip}{-0.4ex}
\addtolength{\belowcaptionskip}{-3.6ex}
}

\begin{document}

\title{Unsupervised and Semi-supervised Anomaly Detection with LSTM Neural Networks}

\author{Tolga Ergen, Ali H. Mirza, and Suleyman S. Kozat \textit{Senior Member, IEEE}\thanks{This work is supported in part by Outstanding Researcher Programme Turkish Academy of Sciences and TUBITAK Contract No 117E153.}
\thanks{The authors are with the Department of Electrical and Electronics Engineering, Bilkent University, Bilkent, Ankara 06800, Turkey, Tel: +90 (312) 290-2336, Fax: +90 (312) 290-1223, {(contact e-mail: \{ergen, mirza, kozat\}@ee.bilkent.edu.tr)}.} }

\maketitle
\begin{abstract}
We investigate anomaly detection in an unsupervised framework and introduce Long Short Term Memory (LSTM) neural network based algorithms. In particular, given variable length data sequences, we first pass these sequences through our LSTM based structure and obtain fixed length sequences. We then find a decision function for our anomaly detectors based on the One Class Support Vector Machines (OC-SVM) and Support Vector Data Description (SVDD) algorithms. As the first time in the literature, we jointly train and optimize the parameters of the LSTM architecture and the OC-SVM (or SVDD) algorithm using highly effective gradient and quadratic programming based training methods. To apply the gradient based training method, we modify the original objective criteria of the OC-SVM and SVDD algorithms, where we prove the convergence of the modified objective criteria to the original criteria. We also provide extensions of our unsupervised formulation to the semi-supervised and fully supervised frameworks. Thus, we obtain anomaly detection algorithms that can process variable length data sequences while providing high performance, especially for time series data. Our approach is generic so that we also apply this approach to the Gated Recurrent Unit (GRU) architecture by directly replacing our LSTM based structure with the GRU based structure. In our experiments, we illustrate significant performance gains achieved by our algorithms with respect to the conventional methods. 
\end{abstract}
\section{Introduction}\label{sec:introduction}
\subsection{Preliminaries}
Anomaly detection \cite{anomalydetection} has attracted significant interest in the contemporary learning literature due its applications in a wide range of engineering problems, e.g., sensor failure \cite{anomaly_sensor}, network monitoring \cite{anomaly_network}, cybersecurity \cite{anomaly_cyber} and surveillance \cite{anomaly_surveillance}. In this paper, we study the variable length anomaly detection problem in an unsupervised framework, where we seek to find a function to decide whether each unlabeled variable length sequence in a given dataset is anomalous or not. Note that although this problem is extensively studied in the literature and there exist different methods, e.g., supervised (or semi-supervised) methods, that require the knowledge of data labels, we employ an unsupervised method due to the high cost of obtaining accurate labels in most real life applications \cite{anomalydetection} such as in cybersecurity \cite{anomaly_cyber} and surveillance \cite{anomaly_surveillance}. However, we also extend our derivations to the semi-supervised and fully supervised frameworks for completeness.

In the current literature, a common and widely used approach for anomaly detection is to find a decision function that defines the model of normality \cite{anomalydetection}. In this approach, one first defines a certain decision function and then optimizes the parameters of this function with respect to a predefined objective criterion, e.g., the  One Class Support Vector Machines (OC-SVM) and Support Vector Data Description (SVDD) algorithms \cite{svm1,svdd}. However, algorithms based on this approach examine time series data over a sufficiently long time window to achieve an acceptable performance \cite{anomalydetection,anomaly_timeseries1, anomaly_timeseries2}. Thus, their performances significantly depend on the length of this time window so that this approach requires careful selection for the length of time window to provide a satisfactory performance \cite{anomaly_lstm,anomaly_timeseries1}. To enhance performance for time series data, neural networks, especially Recurrent Neural Networks (RNNs), based approaches are introduced thanks to their inherent memory structure that can store ``time" or ``state" information \cite{anomalydetection,anomaly_rnn}. However, since the basic RNN architecture does not have control structures (gates) to regulate the amount of information to be stored \cite{rnndisadvantage,lstm}, a more advanced RNN architecture with several control structures, i.e., the Long Short Term Memory (LSTM) network, is introduced \cite{hoch_lstm,lstm}. However, neural networks based approaches do not directly optimize an objective criterion for anomaly detection \cite{anomaly_neural,anomalydetection}. Instead, they first predict a sequence from its past samples and then determine whether the sequence is an anomaly or not based on the prediction error, i.e., an anomaly is an event, which cannot be predicted from the past nominal data \cite{anomalydetection}. Thus, they require a probabilistic model for the prediction error and a threshold on the probabilistic model to detect anomalies, which results in challenging optimization problems and restricts their performance accordingly \cite{anomalydetection,anomaly_neural,anomaly_neural2}. Furthermore, both the common and neural networks based approaches can process only fixed length vector sequences, which significantly limits their usage in real life applications \cite{anomalydetection}.

In order to circumvent these issues, we introduce novel LSTM based anomaly detection algorithms for variable length data sequences. In particular, we first pass variable length data sequences through an LSTM based structure to obtain fixed length representations. We then apply our OC-SVM \cite{svm1} and SVDD \cite{svdd} based algorithms for detecting anomalies in the extracted fixed length vectors as illustrated in Fig. \ref{overallstructure}. Unlike the previous approaches in the literature \cite{anomalydetection}, we jointly train the parameters of the LSTM architecture and the OC-SVM (or SVDD) formulation to maximize the detection performance. For this joint optimization, we propose two different training methods, i.e., a quadratic programming based and a gradient based algorithms, where the merits of each different approach are detailed in the paper. For our gradient based training method, we modify the original OC-SVM and SVDD formulations and then provide the convergence results of the modified formulations to the original ones. Thus, instead of following the prediction based approaches \cite{anomalydetection,anomaly_neural,anomaly_neural2} in the current literature, we define proper objective functions for anomaly detection using the LSTM architecture and optimize the parameters of the LSTM architecture via these well defined objective functions. Hence, our anomaly detection algorithms are able to process variable length sequences and provide high performance for time series data. Furthermore, since we introduce a generic approach in the sense that it can be applied to any RNN architecture, we also apply our approach to the Gated Recurrent Unit (GRU) architecture \cite{gru}, i.e., an advanced RNN architecture as the LSTM architecture, in our simulations. Through extensive set of experiments, we demonstrate significant performance gains with respect to the conventional methods \cite{svm1,svdd}.
\begin{figure}[t]
 \centering 
 \includegraphics[width=.5\textwidth]{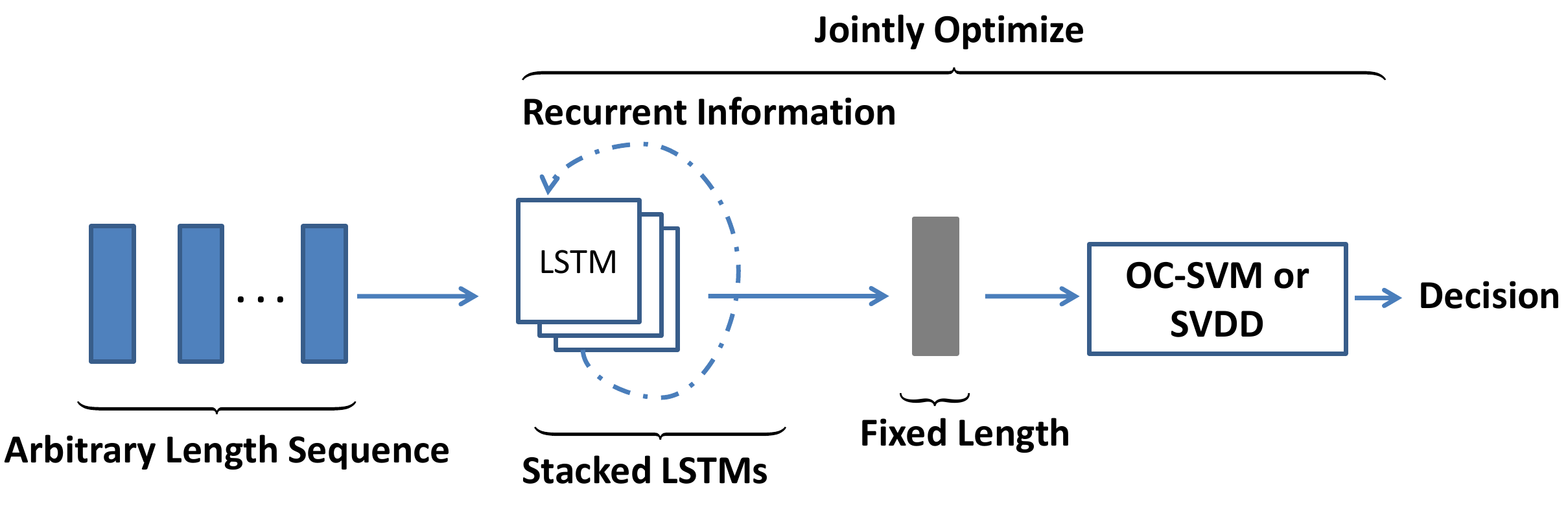}
 \caption{Overall structure of our anomaly detection approach.} \label{overallstructure}
 \end{figure}
\subsection{Prior Art and Comparisons}
Several different methods have been introduced for the anomaly detection problem \cite{anomalydetection}. Among these methods, the OC-SVM \cite{svm1} and SVDD \cite{svdd} algorithms are generally employed due their high performance in real life applications \cite{anomaly_svm_generic}. However, these algorithms provide inadequate performance for time series data due to their inability to capture time dependencies \cite{anomaly_timeseries1,anomaly_timeseries2 }. In order to improve the performances of these algorithms for time series data, in \cite{anomaly_timeseries2}, the authors convert time series data into a set of vectors by replicating each sample so that they obtain two dimensional vector sequences. However, even though they obtain two dimensional vector sequences, the second dimension does not provide additional information such that this approach still provides inadequate performance for time series data \cite{anomaly_timeseries1}. As another approach, the OC-SVM based method in \cite{anomaly_timeseries1} acquires a set of vectors from time series data by unfolding the data into a phase space using a time delay embedding process \cite{timedelay}. More specifically, for a certain sample, they create an $E$ dimensional vector by using the previous $E-1$ samples along with the sample itself \cite{anomaly_timeseries1}. However, in order to obtain a satisfactory performance from this approach, the dimensionality, i.e., $E$, should be carefully tuned, which restricts its usage in real life applications \cite{anomaly_timeseries3}. On the other hand, even though LSTM based algorithms provide high performance for time series data, we have to solve highly complex optimization problems to get an adequate performance \cite{anomalydetection}. As an example, the LSTM based anomaly detection algorithms in \cite{anomaly_lstm,anomaly_lstm_ecg} first predict time series data and then fit a multivariate Gaussian distribution to the error, where they also select a threshold for this distribution. Here, they allocate different set of sequences to learn the parameters of the distribution and threshold via the maximum likelihood estimation technique \cite{anomaly_lstm,anomaly_lstm_ecg}. Thus, the conventional LSTM based approaches require careful selection of several additional parameters, which significantly degrades their performance in real life \cite{anomalydetection,anomaly_lstm}. Furthermore, both the OC-SVM (or SVDD) and LSTM based methods are able to process only fixed length sequences \cite{svm1,svdd,anomaly_lstm}. To circumvent these issues, we introduce generic LSTM based anomaly detectors for variable length data sequences, where we jointly train the parameters of the LSTM architecture and the OC-SVM (or SVDD) formulation via a predefined objective function. Therefore, we not only obtain high performance for time series data but also enjoy joint and effective optimization of the parameters with respect to a well defined objective function.

\subsection{Contributions}
Our main contributions are as follows:
\begin{itemize}
\item We introduce LSTM based anomaly detection algorithms in an unsupervised framework, where we also extend our derivations to the semi-supervised and fully supervised frameworks. 
\item As the first time in the literature, we jointly train the parameters of the LSTM architecture and the OC-SVM (or SVDD) formulation via a well defined objective function, where we introduce two different joint optimization methods. For our gradient based joint optimization method, we modify the OC-SVM and SVDD formulations and then prove the convergence of the modified formulations to the original ones.
\item Thanks to our LSTM based structure, the introduced methods are able to process variable length data sequences. Additionally, unlike the conventional methods \cite{svm1,svdd}, our methods effectively detect anomalies in time series data without requiring any preprocessing.
\item Through extensive set of experiments involving real and simulated data, we illustrate significant performance improvements achieved by our algorithms with respect to the conventional methods \cite{svm1,svdd}. Moreover, since our approach is generic, we also apply it to the recently proposed GRU architecture \cite{gru} in our experiments.
\end{itemize}
\subsection{Organization of the Paper}
The organization of this paper is as follows. In Section \ref{sec:problem_description}, we first describe the variable length anomaly detection problem and then introduce our LSTM based structure. In Section \ref{sec:sub_svm}, we introduce anomaly detection algorithms based on the OC-SVM formulation, where we also propose two different joint training methods in order to learn the LSTM and SVM parameters. The merits of each different approach are also detailed in the same section. In a similar manner, we introduce anomaly detection algorithms based on the SVDD formulation and provide two different joint training methods to learn the parameters in Section \ref{sec:sub_svdd}. In Section \ref{sec:simulations}, we demonstrate performance improvements over several real life datasets. In the same section, thanks to our generic approach, we also introduce GRU based anomaly detection algorithms. Finally, we provide concluding remarks in Section \ref{sec:conclusion}.

\section{Model and Problem Description} \label{sec:problem_description}
In this paper, all vectors are column vectors and denoted by boldface
lower case letters. Matrices are represented by boldface uppercase
letters. For a vector $\vec{a}$, $\vec{a}^T$ is its
ordinary transpose and $\norm{\vec{a}} =
\sqrt{\vec{a}^T\vec{a}}$ is the $\ell^2$-norm. The time index is given as subscript, e.g., $\vec{a}_i$ is the $i$\textsuperscript{th} vector. Here, $\vec{1}$ (and $\vec{0}$) is a vector of all ones (and zeros) and $\vec{I}$ represents the identity matrix, where the sizes are understood from the context.

We observe data sequences $\lbrace \vec{X}_i \rbrace_{i=1}^n$, i.e., defined as
\begin{align*}
\vec{X}_{i}=[\vec{x}_{i,1} \text{ } \vec{x}_{i,2}\ldots \vec{x}_{i,d_i} ],
\end{align*}
where $\vec{x}_{i,j} \in \mathbb{R}^p$, $\forall j \in \lbrace 1,2, \ldots d_i\rbrace$ and $d_i \in \mathbb{Z}^{+}$ is the number of columns in $\vec{X}_i$, which can vary with respect to $i$. Here, we assume that the bulk of the observed sequences are normal and the remaining sequences are anomalous. Our aim is to find a scoring (or decision) function to determine whether $\vec{X}_i$ is anomalous or not based on the observed data, where $+1$ and $-1$ represent the outputs of the desired scoring function for nominal and anomalous data respectively. As an example application for this framework, in host based intrusion detection \cite{anomalydetection}, the system handles operating system call traces, where the data consists of system calls that are generated by users or programs. All traces contain system calls that belong to the same alphabet, however, the co-occurrence of the system calls is the key issue in detecting anomalies \cite{anomalydetection}. For different programs, these system calls are executed in different sequences, where the length of the sequence may vary for each program. Binary encoding of a sample set of call sequences can be $\vec{X}_1=101011$, $\vec{X}_2=1010$ and $\vec{X}_3=1011001$ for $n=3$ case \cite{anomalydetection}. After observing such a set of call sequences, our aim is to find a scoring function that successfully distinguishes the anomalous call sequences from the normal sequences.

In order to find a scoring function $l(\cdot)$ such that
\begin{align*}
l(\vec{X}_i) =
\begin{cases} 
-1 & \text{ if } \vec{X}_i \text{ is anomalous}\\    +1 & \text{ otherwise }
\end{cases},
\end{align*}
one can use the OC-SVM algorithm \cite{svm1} to find a hyperplane that separates the anomalies from the normal data or the SVDD algorithm \cite{svdd} to find a hypersphere enclosing the normal data while leaving the anomalies outside the hypersphere. However, these algorithms can only process fixed length sequences. Hence, we use the LSTM architecture \cite{hoch_lstm} to obtain a fixed length vector representation for each $\vec{X_i}$. Although there exist several different versions of LSTM architecture, we use the most widely employed architecture, i.e., the LSTM architecture without peephole connections \cite{lstm}. We first feed $\vec{X}_i$ to the LSTM architecture as demonstrated in Fig. \ref{pooling}, where the internal LSTM equations are as follows \cite{hoch_lstm}:
\begin{align}
\label{eq:z}
&\vec{z}_{i,j}=g(\vec{W}^{(z)}\vec{x}_{i,j}+\vec{R}^{(z)}\vec{h}_{i,j-1}+\vec{b}^{(z)})\\
  & \vec{s}_{i,j} = \sigma(\vec{W}^{(s)}\vec{x}_{i,j}+\vec{R}^{(s)}\vec{h}_{i,j-1}+\vec{b}^{(s)}) \label{eq:s} \\
    & \vec{f}_{i,j} = \sigma(\vec{W}^{(f)}\vec{x}_{i,j}+\vec{R}^{(f)}\vec{h}_{i,j-1}+\vec{b}^{(f)}) \label{eq:f} \\
  & \vec{c}_{i,j} = \vec{s}_{i,j} \odot \vec{z}_{i,j}+ \vec{f}_{i,j}\odot \vec{c}_{i,j-1} \label{eq:state} \\
  & \vec{o}_{i,j}= \sigma(\vec{W}^{(o)}\vec{x}_{i,j}+\vec{R}^{(o)}\vec{h}_{i,j-1}+ \vec{b}^{(o)}) \label{eq:o} \\
  & \vec{h}_{i,j}= \vec{o}_{i,j}\odot g(\vec{c}_{i,j}), \label{eq:output}
\end{align}
where $\vec{c}_{i,j} \in \mathbb{R}^m$ is the state vector, $\vec{x}_{i,j} \in \mathbb{R}^p$ is the input vector and $\vec{h}_{i,j}\in \mathbb{R}^m$ is the output vector for the $j$\textsuperscript{th} LSTM unit in Fig. \ref{pooling}. Additionally, $\vec{s}_{i,j}$, $\vec{f}_{i,j}$ and $\vec{o}_{i,j}$ is the input, forget and output gates, respectively. Here, $g(\cdot)$ is set to the hyperbolic tangent function, i.e., $\tanh$, and applies to input vectors pointwise. Similarly, $\sigma(\cdot)$ is set to the sigmoid function. $\odot$ is the operation for elementwise multiplication of two same sized vectors. Furthermore, $\vec{W}^{(\cdot)}$, $\vec{R}^{(\cdot)}$ and $\vec{b}^{(\cdot)}$ are the parameters of the LSTM architecture, where the size of each is selected according to the dimensionality of the input and output vectors. After applying the LSTM architecture to each column of our data sequences as illustrated in Fig. \ref{pooling}, we take the average of the LSTM outputs for each data sequence, i.e., the mean pooling method. By this, we obtain a new set of fixed length sequences, i.e., denoted as $\lbrace \vec{\bar{h}}_i \rbrace_{i=1}^n$, $\vec{\bar{h}}_i \in \mathbb{R}^m$. Note that we also use the same procedure to obtain the state information $\vec{\bar{c}}_i \in \mathbb{R}^m$ for each $\vec{X}_i$ as demonstrated in Fig. \ref{pooling}.
\begin{figure}[t]
 \centering 
 \includegraphics[width=.45\textwidth]{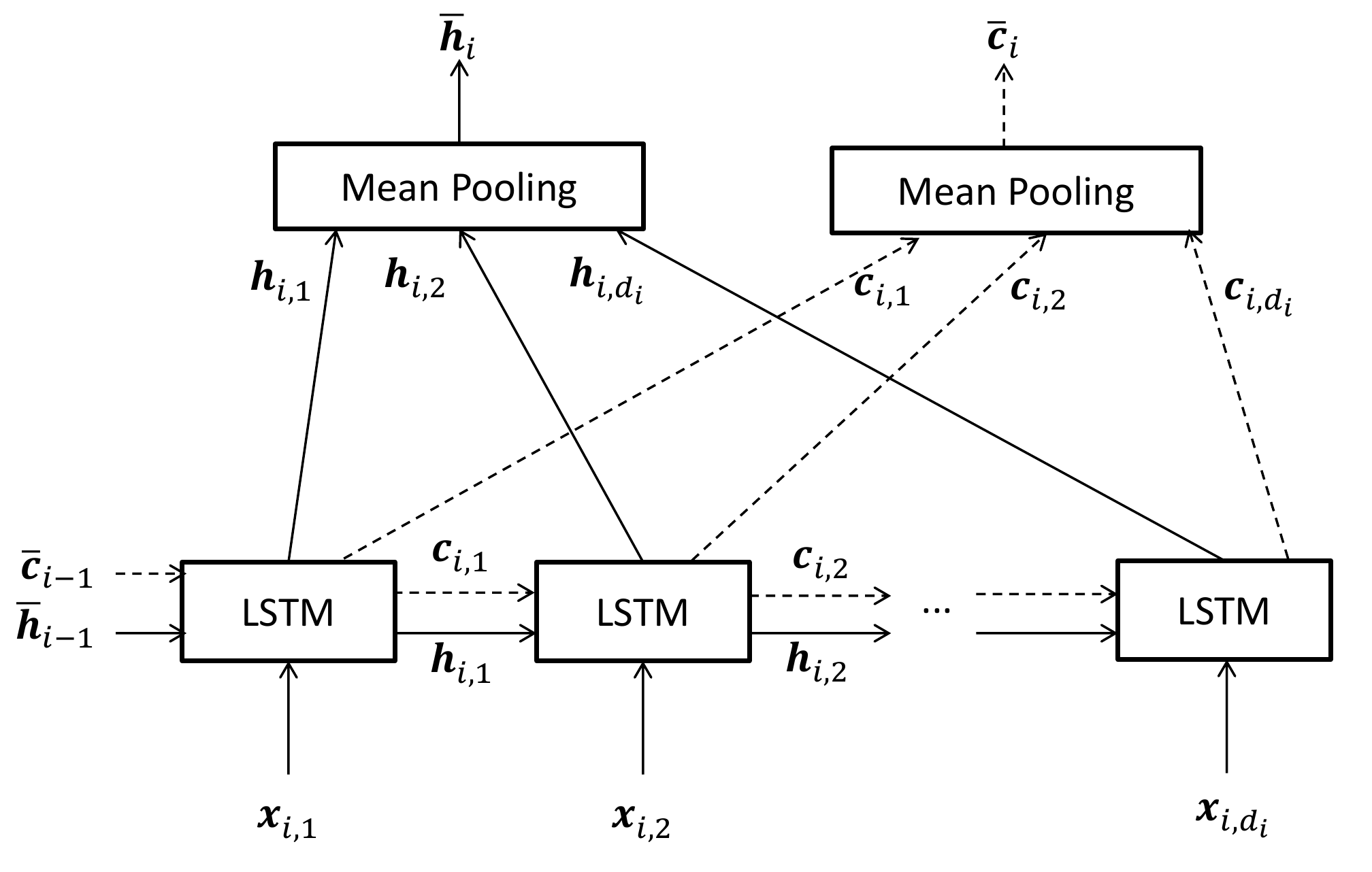}\\
 \caption{Our LSTM based structure for obtaining fixed length sequences.} \label{pooling}
\end{figure}
\begin{remark}
We use the mean pooling method in order to obtain the fixed length sequences as 
$
\vec{\bar{h}}_i= \frac{1}{d_i} \sum_{j=1}^{d_i} \vec{h}_{i,j}.
$
However, we can also use the other pooling methods. As an example, for the last and max pooling methods, we use $\vec{\bar{h}}_{i}=\vec{h}_{i,d_i}$ and $\vec{\bar{h}}_{i}=\max_{j} \vec{h}_{i,j}$, $\forall i \in \lbrace 1, 2, \ldots n \rbrace$, respectively. Our derivations can be straightforwardly extended to these different pooling methods.
\end{remark}

\section{Novel Anomaly Detection Algorithms}\label{sec:anomaly_detection}
In this section, we first formulate the anomaly detection approaches based on the OC-SVM and SVDD algorithms. We then provide joint optimization updates to train the parameters of the overall structure.
\subsection{Anomaly Detection with the OC-SVM Algorithm} \label{sec:sub_svm}
In this subsection, we provide an anomaly detection algorithm based on the OC-SVM formulation and derive the joint updates for both the LSTM and SVM parameters. For the training, we first provide a quadratic programming based algorithm and then introduce a gradient based training algorithm. To apply the gradient based training method, we smoothly approximate the original OC-SVM formulation and then prove the convergence of the approximated formulation to the actual one in the following subsections.

In the OC-SVM algorithm, our aim is to find a hyperplane that separates the anomalies from the normal data \cite{svm1}. We formulate the OC-SVM optimization problem for the sequence $\lbrace \vec{\bar{h}}_i \rbrace_{i=1}^n$ as follows \cite{svm1}
\begin{align}
&\min_{\mbox{\scriptsize\vec{\theta}} \in \mathbb{R}^{n_{\theta}},\mbox{\scriptsize\vec{w}} \in \mathbb{R}^m,\xi \in \mathbb{R}, \rho \in \mathbb{R}}\text{    }\frac{\|\vec{w}\|^{2}}{2}+\frac{1}{n\lambda} \sum_{i=1}^{n} \xi_{i}-\rho \label{svm_primal}\\ 
 &\text{ subject to: } \vec{w}^{T}\vec{\bar{h}}_{i} \geq \rho-\xi_{i} \text{, } \xi_{i}\geq0  \text{, } \forall i \label{svm_primal_constraint}\\
 &\vec{W}^{(\cdot)T}\vec{W}^{(\cdot)}=\vec{I},  \vec{R}^{(\cdot)T}\vec{R}^{(\cdot)}=\vec{I} \text{ and } \vec{b}^{(\cdot)T}\vec{b}^{(\cdot)}=1 \label{svm_primal_constraint2},
\end{align}
where $\rho$ and $\vec{w}$ are the parameters of the separating hyperplane, $\lambda >0$ is a regularization parameter, $\xi$ is a slack variable to penalize misclassified instances and we group the LSTM parameters {\def\OldComma{,}
    \catcode`\,=13
    \def,{%
      \ifmmode%
        \OldComma\discretionary{}{}{}%
      \else%
        \OldComma%
      \fi  }%
    $\lbrace\vec{W}^{(z)}, \vec{R}^{(z)}, \vec{b}^{(z)}, \vec{W}^{(s)}, \vec{R}^{(s)}, \vec{b}^{(s)}, \vec{W}^{(f)}, \vec{R}^{(f)}, \vec{b}^{(f)}, \vec{W}^{(o)}, \vec{R}^{(o)}, \vec{b}^{(o)} \rbrace$} into $\vec{\theta} \in \mathbb{R}^{n_{\theta}}$, where $n_{\theta}=4m(m+p+1)$. Since the LSTM parameters are unknown and $\vec{\bar{h}}_i$ is a function of these parameters, we also minimize the cost function in \eqref{svm_primal} with respect to $\vec{\theta}$.  

After solving the optimization problem in \eqref{svm_primal}, \eqref{svm_primal_constraint} and \eqref{svm_primal_constraint2}, we use the scoring function
\begin{align}
\label{decision_svm}
l(\vec{X}_i) =\sgn(\vec{w}^T\vec{\bar{h}}_i-\rho)
\end{align}
to detect the anomalous data, where the $\sgn(\cdot)$ function returns the sign of its input.

We emphasize that while minimizing \eqref{svm_primal} with respect to $\vec{\theta}$, we might suffer from overfitting and impotent learning of time dependencies on the data \cite{goru}, i.e., forcing the parameters to null values, e.g., $\vec{\theta}=\vec{0}$. To circumvent these issues, we introduce \eqref{svm_primal_constraint2}, which constraints the norm of $\vec{\theta}$ to avoid overfitting and trivial solutions, e.g., $\vec{\theta}=\vec{0}$, while boosting the ability of the LSTM architecture to capture time dependencies \cite{goru,rnn_unitary}. 
\begin{remark}
\label{lstm_constraint}
In \eqref{svm_primal_constraint2}, we use an orthogonality constraint for each LSTM parameter. However, we can also use other constraints instead of \eqref{svm_primal_constraint2} and solve the optimization problem in \eqref{svm_primal}, \eqref{svm_primal_constraint} and \eqref{svm_primal_constraint2} in the same manner. As an example, a common choice of constraint for neural networks is the Frobenius norm \cite{frobenius}, i.e., defined as
\begin{align}
\label{frobenius}
\| \vec{A} \|_F=\sum_{i} \sum_j \vec{A}_{ij}^2
\end{align}
for a real matrix $\vec{A}$, where $\vec{A}_{ij}$ represents the element at the $i$\textsuperscript{th} column and $j$\textsuperscript{th} row of $\vec{A}$. In this case, we can directly replace \eqref{svm_primal_constraint2} with a Frobenius norm constraint for each LSTM parameter as in \eqref{frobenius} and then solve the optimization problem in the same manner. Such approaches only aim to regularize the parameters \cite{rnn_unitary}. However, for RNNs, we may also encounter exponential growth or decay in the norm of the gradients while training the parameters, which significantly degrades capabilities of these architectures to capture time dependencies \cite{goru,rnn_unitary}. Thus, in this paper, we put the constraint \eqref{svm_primal_constraint2} in order to regularize the parameters while improving the capabilities of the LSTM architecture in capturing time dependencies \cite{goru,rnn_unitary}.
\end{remark}
 \subsubsection{Quadratic Programming Based Training Algorithm}
Here, we introduce a training approach based on quadratic programming for the optimization problem in \eqref{svm_primal}, \eqref{svm_primal_constraint} and \eqref{svm_primal_constraint2}, where we perform consecutive updates for the LSTM and SVM parameters. For this purpose, we first convert the optimization problem to a dual form in the following. We then provide the consecutive updates for each parameter.

We have the following Lagrangian for the SVM parameters
\begin{align}
L(\vec{w},\xi,\rho,\nu,\alpha)=&\frac{\|\vec{w}\|^{2}}{2}+\frac{1}{n\lambda} \sum_{i=1}^{n} \xi_{i}-\rho-\sum_{i=1}^{n} \nu_{i}\xi_{i} \nonumber \\
&-\sum_{i=1}^{n} \alpha_{i}(\vec{w}^{T}\vec{\bar{h}}_{i}-\rho+\xi_{i})\label{lagrange_svm_quad},
\end{align}
where $\nu_{i}$, $\alpha_{i} \geq 0$ are the Lagrange multipliers. Taking derivative of \eqref{lagrange_svm_quad} with respect to $\vec{w}$, $\xi$ and $\rho$ and then setting the derivatives to zero gives
\begin{align}
&\vec{w}=\sum_{i=1}^{n}\alpha_{i}\vec{\bar{h}}_{i} \label{wlagrange_quad}\\
&\sum_{i=1}^{n}\alpha_{i}=1 \text{ and } \alpha_{i}=1/(n\lambda)-\nu_{i} \text{, } \forall i. \label{lagrange_quad}
\end{align}
Note that at the optimum, the inequalities in \eqref{svm_primal_constraint} become equalities if $\alpha_i$ and $\nu_i$ are nonzero, i.e.,  $0<\alpha_i <1/(n\lambda)$ \cite{svm1}. With this relation, we compute $\rho$ as
\begin{align}
\rho=\sum_{j=1}^{n}\alpha_{j}\vec{\bar{h}}_j^T\vec{\bar{h}}_i \text{ for } 0<\alpha_i<1/(n\lambda). \label{plagrange_quad}
\end{align}

By substituting \eqref{wlagrange_quad} and \eqref{lagrange_quad} into \eqref{lagrange_svm_quad}, we obtain the following dual problem for the constrained minimization in \eqref{svm_primal}, \eqref{svm_primal_constraint} and \eqref{svm_primal_constraint2}
\begin{align}
&\min_{\mbox{\scriptsize\vec{\theta}} \in \mathbb{R}^{n_{\theta}},\mbox{\scriptsize\vec{\alpha}} \in \mathbb{R}^n}\frac{1}{2}\sum_{i=1}^{n} \sum_{j=1}^{n} \alpha_{i}\alpha_{j} \vec{\bar{h}}_i^T \vec{\bar{h}}_j \label{svm_dual}\\
&\text{subject to: } \sum_{i=1}^{n}\alpha_{i}=1 \text{ and }\ 0\leq \alpha_{i} \leq 1/(n\lambda) \text{, }\forall i \label{svm_dual_constraint}\\
 &\vec{W}^{(\cdot)T}\vec{W}^{(\cdot)}=\vec{I},  \vec{R}^{(\cdot)T}\vec{R}^{(\cdot)}=\vec{I} \text{ and } \vec{b}^{(\cdot)T}\vec{b}^{(\cdot)}=1, \label{svm_dual_constraint2}
\end{align} 
where $\vec{\alpha}\in \mathbb{R}^n$ is a vector representation for $\alpha_i$'s. Since the LSTM parameters are unknown, we also put the minimization term for $\vec{\theta}$ into \eqref{svm_dual} as in \eqref{svm_primal}.
By substituting \eqref{wlagrange_quad} into \eqref{decision_svm}, we have the following scoring function for the dual problem
\begin{align}
\label{svm_dual_decision}
l(\vec{X}_i)=\sgn \big(\sum_{j=1}^{n}\alpha_{j}\vec{\bar{h}}_j^T\vec{\bar{h}}_i-\rho \big) ,
\end{align}
where we calculate $\rho$ using \eqref{plagrange_quad}.

In order to find the optimal $\vec{\theta}$ and $\vec{\alpha}$ for the optimization problem in \eqref{svm_dual}, \eqref{svm_dual_constraint} and \eqref{svm_dual_constraint2}, we employ the following procedure. We first select a certain set of the LSTM parameters, i.e., $\vec{\theta}_0$. Based on $\vec{\theta}_0$, we find the minimizing $\vec{\alpha}$ values, i.e., $\vec{\alpha}_1$, using the Sequential Minimal Optimization (SMO) algorithm \cite{smo}. Now, we fix $\vec{\alpha}$ as $\vec{\alpha}_1$ and then update $\vec{\theta}$ from $\vec{\theta}_0$ to $\vec{\theta}_1$ using the algorithm for optimization with orthogonality constraints in \cite{feasiblemethod}. We repeat these consecutive update procedures until $\vec{\alpha}$ and $\vec{\theta}$ converge \cite{remez}. Then, we use the converged values in order to evaluate \eqref{svm_dual_decision}. In the following, we explain the update procedures for $\vec{\theta}$ and $\vec{\alpha}$ in detail.

Based on $\vec{\theta}_k$, i.e., the LSTM parameter vector at the $k$\textsuperscript{th} iteration, we update $\vec{\alpha}_k$, i.e., the $\vec{\alpha}$ vector at the $k$\textsuperscript{th} iteration, using the SMO algorithm due to its efficiency in solving quadratic constrained optimization problems \cite{smo}. In the SMO algorithm, we choose a subset of parameters to minimize and fix the rest of parameters. In the extreme case, we choose only one parameter to minimize, however, due to \eqref{svm_dual_constraint}, we must choose at least two parameters. To illustrate how the SMO algorithm works in our case, we choose $\alpha_{1}$ and $\alpha_{2}$ to update and fix the rest of the parameters in \eqref{svm_dual}. From \eqref{svm_dual_constraint}, we have
\begin{align}
\alpha_{1}=1-S-\alpha_{2}, \text{ where } S=\sum_{i=3}^{n}\alpha_{i}. \label{smo1}
\end{align} 
We first replace $\alpha_{1}$ in \eqref{svm_dual} with \eqref{smo1}. We then take the derivative of \eqref{svm_dual} with respect to $\alpha_{2}$ and equate the derivative to zero. Thus, we obtain the following update for $\alpha_{2}$ at the $k$\textsuperscript{th} iteration
\begin{align}
\alpha_{k+1,2}=\frac{(\alpha_{k,1}+\alpha_{k,2})(K_{11}-K_{12})+M_1 - M_2}{K_{11}+K_{22}-2K_{12}},\label{smo2}
\end{align}
where $K_{ij} \triangleq\vec{\bar{h}}_i^T \vec{\bar{h}}_j$, $M_i\triangleq\sum_{j=3}^n \alpha_{k,j} K_{ij}$ and $\alpha_{k,i}$ represents the $i$\textsuperscript{th} element of $\vec{\alpha}_k$. Due to \eqref{svm_dual_constraint}, if the updated value of $\alpha_2$ is outside of the region $[0,1/(n\lambda)]$, we project it to this region. Once $\alpha_2$ is updated as $\alpha_{k+1,2}$, we obtain $\alpha_{k+1,1}$ using \eqref{smo1}. For the rest of the parameters, we repeat the same procedure, which eventually converges to a certain set of parameters \cite{smo}. By this way, we obtain $\vec{\alpha}_{k+1}$, i.e., the minimizing $\vec{\alpha}$ for $\vec{\theta}_k$.

Following the update of $\vec{\alpha}$, we update $\vec{\theta}$ based on the updated $\vec{\alpha}_{k+1}$ vector. For this purpose, we employ the optimization method in \cite{feasiblemethod}. Since we have $\vec{\alpha}_{k+1}$ that satisfies \eqref{svm_dual_constraint}, we reduce the dual problem to 
\begin{align}
\label{svm_dual_thetamin}
&\min_{\mbox{\scriptsize\vec{\theta}}} \text{ }\kappa(\vec{\theta},\vec{\alpha}_{k+1})= \frac{1}{2}\sum_{i=1}^{n} \sum_{j=1}^{n} \alpha_{k+1,i} \alpha_{k+1,j}  \vec{\bar{h}}_i^T \vec{\bar{h}}_j\\
 & \text{s.t.:}\vec{W}^{(\cdot)T}\vec{W}^{(\cdot)}=\vec{I},  \vec{R}^{(\cdot)T}\vec{R}^{(\cdot)}=\vec{I} \text{ and } \vec{b}^{(\cdot)T}\vec{b}^{(\cdot)}=1.\label{svm_dual_thetamincons}
\end{align}
For \eqref{svm_dual_thetamin} and \eqref{svm_dual_thetamincons}, we update $\vec{W}^{(\cdot)}$ as follows
\begin{align}
\label{svm_dual_Wupdate}
\vec{W}^{(\cdot)}_{k+1}= \bigg( \vec{I}+\frac{\mu}{2} \vec{A}_k \bigg)^{-1} \bigg( \vec{I}-\frac{\mu}{2} \vec{A}_k \bigg) \vec{W}^{(\cdot)}_{k},
\end{align}
where the subscripts represent the current iteration index, $\mu$ is the learning rate, $\vec{A}_k=\vec{G}_k (\vec{W}^{(\cdot)}_k)^{T}-\vec{W}^{(\cdot)}_k\vec{G}_k^T$ and the element at the $i$\textsuperscript{th} row and the $j$\textsuperscript{th} column of $\vec{G}$, i.e., $\vec{G}_{ij}$, is defined as
\begin{align}
\label{derivative_W}\vec{G}_{ij}\triangleq \frac{\partial \kappa(\vec{\theta},\vec{\alpha}_{k+1})}{\partial \vec{W}^{(\cdot)}_{ij}}.
\end{align}
\begin{remark} 
\label{remark2}
For $\vec{R}^{(\cdot)}$ and $\vec{b}^{(\cdot)}$, we first compute the gradient of the objective function with respect to the chosen parameter as in \eqref{derivative_W}. We then obtain $\vec{A}_k$ according to the chosen parameter. Using $\vec{A}_k$, we update the chosen parameter as in \eqref{svm_dual_Wupdate}.
\end{remark}
With these updates, we obtain a quadratic programming based training algorithm (see Algorithm \ref{alg1} for the pseudocode) for our LSTM based anomaly detector.
\begin{algorithm}
\caption{Quadratic Programming Based Training for the Anomaly Detection Algorithm Based on OC-SVM}
\label{alg1}
\begin{algorithmic}[1]
\State{Initialize the LSTM parameters as $\vec{\theta}_0$ and the dual OC-SVM parameters as $\vec{\alpha}_0$}
\State{Determine a threshold $\epsilon$ as convergence criterion}
\State{$k=-1$}
\Do
\State{$k=k+1$}
\State{Using $\vec{\theta}_k$, obtain $\lbrace \vec{\bar{h}}\rbrace_{i=1}^n$ according to Fig. \ref{pooling}}
\State{Find optimal $\vec{\alpha}_{k+1}$ for $\lbrace \vec{\bar{h}}\rbrace_{i=1}^n$ using \eqref{smo1} and \eqref{smo2}}
\State{Based on $\vec{\alpha}_{k+1}$, obtain $\vec{\theta}_{k+1} $ using \eqref{svm_dual_Wupdate} and Remark \ref{remark2}}
\doWhile{$\big( \kappa(\vec{\theta}_{k+1},\vec{\alpha}_{k+1}) -\kappa(\vec{\theta}_k,\vec{\alpha}_{k})\big)^2 >\epsilon$}
\State{Detect anomalies using \eqref{svm_dual_decision} evaluated at $\vec{\theta}_k$ and $\vec{\alpha}_k$ }
\end{algorithmic}
\end{algorithm}
\subsubsection{Gradient Based Training Algorithm}
Although the quadratic programming based training algorithm directly optimizes the original OC-SVM formulation without requiring any approximation, since it depends on the separated consecutive updates of the LSTM and OC-SVM parameters, it might not converge to even a local minimum \cite{remez}. In order to resolve this issue, in this subsection, we introduce a training method based on only the first order gradients, which updates the parameters at the same time. However, since we require an approximation to the original OC-SVM formulation to apply this method, we also prove the convergence of the approximated formulation to the original OC-SVM formulation in this subsection.

Considering \eqref{svm_primal_constraint}, we write the slack variable in a different form as follows
 \begin{align}
\label{hinge_svm}
G(\beta_{\mbox{\scriptsize\vec{w}},\rho}(\vec{\bar{h}}_i))\triangleq\max\lbrace0,\beta_{\mbox{\scriptsize\vec{w}},\rho}(\vec{\bar{h}}_i) \rbrace, \forall i,
\end{align}
where
\begin{align*}
\beta_{\mbox{\scriptsize\vec{w}},\rho}(\vec{\bar{h}}_i) \triangleq \rho-\vec{w}^{T}\vec{\bar{h}}_{i}.
\end{align*}
By substituting \eqref{hinge_svm} into \eqref{svm_primal}, we remove the constraint \eqref{svm_primal_constraint} and obtain the following optimization problem
\begin{align}
\label{svm_unconstrained}
&\min_{\mbox{\scriptsize\vec{w}}\in \mathbb{R}^m, \rho \in \mathbb{R}, \mbox{\scriptsize\vec{\theta}}\in \mathbb{R}^{n_{\theta}}}\text{   }\frac{\|\vec{w}\|^{2}}{2}+\frac{1}{n \lambda} \sum_{i=1}^{n} G(\beta_{\mbox{\scriptsize\vec{w}},\rho}(\vec{\bar{h}}_i))- \rho\\ \label{svm_unconstrained_constrained}
& \text{s.t.:}\vec{W}^{(\cdot)T}\vec{W}^{(\cdot)}=\vec{I},  \vec{R}^{(\cdot)T}\vec{R}^{(\cdot)}=\vec{I} \text{ and } \vec{b}^{(\cdot)T}\vec{b}^{(\cdot)}=1.
\end{align}
Since \eqref{hinge_svm} is not a differentiable function, we are unable to solve the optimization problem in \eqref{svm_unconstrained} using gradient based optimization algorithms. Hence, we employ a differentiable function
\begin{align}
\label{hinge_approx}
S_{\tau}(\beta_{\mbox{\scriptsize\vec{w}},\rho}(\vec{\bar{h}}_i)  )=\frac{1}{\tau}\log\bigg(1+e^{\tau \beta_{\mbox{\scriptsize\vec{w}},\rho}(\vec{\bar{\mbox{\scriptsize h}}}_i) }\bigg)
\end{align}
to smoothly approximate \eqref{hinge_svm}, where $\tau>0$ is the smoothing parameter and $\log$ represents the natural logarithm. In \eqref{hinge_approx}, as $\tau$ increases, $S_{\tau}(\cdot)$ converges to $G(\cdot)$ (see Proposition \ref{proposition1} at the end of this section), hence, we choose a large value for $\tau$. With \eqref{hinge_approx}, we modify our optimization problem as follows
\begin{align}
\label{svm_unconstrained_approx}
&\min_{\mbox{\scriptsize\vec{w}}\in \mathbb{R}^m, \rho \in \mathbb{R}, \mbox{\scriptsize\vec{\theta}}\in \mathbb{R}^{n_{\theta}}}\text{}F_{\tau}(\vec{w},\rho,\vec{\theta})\\
 \label{svm_unconstrained_approx_constrained}
& \text{s.t.:}\vec{W}^{(\cdot)T}\vec{W}^{(\cdot)}=\vec{I},  \vec{R}^{(\cdot)T}\vec{R}^{(\cdot)}=\vec{I} \text{ and } \vec{b}^{(\cdot)T}\vec{b}^{(\cdot)}=1
\end{align}
where $F_{\tau}(\cdot,\cdot, \cdot)$ is the objective function of our optimization problem and defined as 
\begin{align*}
F_{\tau}(\vec{w},\rho,\vec{\theta})\triangleq\frac{\|\vec{w}\|^{2}}{2}+\frac{1}{n \lambda} \sum_{i=1}^{n} S_{\tau}(\beta_{\mbox{\scriptsize\vec{w}},\rho}(\vec{\bar{h}}_i)  )-\rho.
\end{align*}
To obtain the optimal parameters for \eqref{svm_unconstrained_approx} and \eqref{svm_unconstrained_approx_constrained}, we update $\vec{w}$, $\rho$ and $\vec{\theta}$ until they converge to a local or global optimum \cite{sayed,feasiblemethod}. For the update of $\vec{w}$ and $\rho$, we use the SGD algorithm \cite{sayed}, where we compute the first order gradient of the objective function with respect to each parameter. We first compute the gradient for $\vec{w}$ as follows
\begin{align}
\label{gradient_w}
\nabla_{\mbox{\scriptsize \vec{w}}}F_{\tau}(\vec{w},\rho,\vec{\theta})=\vec{w}+\frac{1}{n\lambda }\sum_{i=1}^n\frac{-\vec{\bar{h}}_ie^{\tau \beta_{\mbox{\scriptsize\vec{w}},\rho}(\vec{\bar{\mbox{\scriptsize h}}}_i) }}{1+e^{\tau \beta_{\mbox{\scriptsize\vec{w}},\rho}(\vec{\bar{\mbox{\scriptsize h}}}_i) }}.
\end{align}
Using \eqref{gradient_w}, we update $\vec{w}$ as 
\begin{align}
\label{update_w}
\vec{w}_{k+1}=\vec{w}_{k}-\mu\nabla_{\mbox{\scriptsize \vec{w}}}F_{\tau}(\vec{w},\rho,\vec{\theta})\Bigr|_{\scriptsize \substack{  \vec{w}=\vec{w}_{k} \\ \scalebox{1.1}{$\rho$}=\scalebox{1.1}{$\rho_{k}$} 
\\  \vec{\theta}=\vec{\theta}_{k}}},
\end{align}
where the subscript $k$ indicates the value of any parameter at the $k$\textsuperscript{th} iteration. Similarly, we calculate the derivative of the objective function with respect to $\rho$ as follows
\begin{align}
\label{gradient_rho}
\frac{\partial F_{\tau}(\vec{w},\rho,\vec{\theta})}{\partial \rho}=\frac{1}{n\lambda }\sum_{i=1}^n\frac{ e^{\tau \beta_{\mbox{\scriptsize\vec{w}},\rho}(\vec{\bar{\mbox{\scriptsize h}}}_i) }}{1+e^{\tau \beta_{\mbox{\scriptsize\vec{w}},\rho}(\vec{\bar{\mbox{\scriptsize h}}}_i) }}-1.
\end{align}
Using \eqref{gradient_rho}, we update $\rho$ as 
\begin{align}
\label{update_rho}
\rho_{k+1}=\rho_{k}-\mu\frac{\partial F_{\tau}(\vec{w},\rho,\vec{\theta})}{\partial \rho}\Bigr|_{\scriptsize \substack{  \vec{w}=\vec{w}_{k} \\ \scalebox{1.1}{$\rho$}=\scalebox{1.1}{$\rho_{k}$} 
\\  \vec{\theta}=\vec{\theta}_{k}}}.
\end{align}
For the LSTM parameters, we use the method for optimization with orthogonality constraints in \cite{feasiblemethod} due to \eqref{svm_unconstrained_approx_constrained}. To update each element of $\vec{W}^{(\cdot)}$, we calculate the gradient of the objective function as
\begin{align}
\label{svm_gradient_W}
\frac{\partial F_{\tau}(\vec{w},\rho,\vec{\theta})}{\partial \vec{W}^{(\cdot)}_{ij}}=\frac{1}{n\lambda }\sum_{i=1}^n\frac{- \vec{w}^T \big(\partial \vec{\bar{h}}_i/\partial \vec{W}^{(\cdot)}_{ij}\big)  e^{\tau \beta_{\mbox{\scriptsize\vec{w}},\rho}(\vec{\bar{\mbox{\scriptsize h}}}_i) }}{1+e^{\tau \beta_{\mbox{\scriptsize\vec{w}},\rho}(\vec{\bar{\mbox{\scriptsize h}}}_i) }}.
\end{align}
We then update $\vec{W}^{(\cdot)}$ using \eqref{svm_gradient_W} as
\begin{align}
\label{svm_gradient_Wupdate}
\vec{W}_{k+1}^{(\cdot)}= \bigg( \vec{I}+\frac{\mu}{2} \vec{B}_{k} \bigg)^{-1} \bigg( \vec{I}-\frac{\mu}{2} \vec{B}_{k} \bigg)\vec{W}_{k}^{(\cdot)},
\end{align}
where $\vec{B}_k=\vec{M}_k (\vec{W}^{(\cdot)}_k)^{T}-\vec{W}^{(\cdot)}_k\vec{M}_k^T$ and 
\begin{align}
\label{svm_gradient_Wgradient}
\vec{M}_{ij} \triangleq\frac{ \partial F_{\tau}(\vec{w},\rho,\vec{\theta})}{\partial \vec{W}^{(\cdot)}_{ij}}.
\end{align}
\begin{remark}
\label{remark_svm_gradient}
For $\vec{R}^{(\cdot)}$ and $\vec{b}^{(\cdot)}$, we first compute the gradient of the objective function with respect to the chosen parameter as in \eqref{svm_gradient_Wgradient}. We then obtain $\vec{B}_k$ according to the chosen parameter. Using $\vec{B}_k$, we update the chosen parameter as in \eqref{svm_gradient_Wupdate}.
\end{remark}
\begin{remark}
\label{svm_semi}
In the semi-supervised framework, we have the following optimization problem for our SVM based algorithms \cite{svm_semisupervised}
\begin{align}
&\min_{\mbox{\scriptsize\vec{\theta}} ,\mbox{\scriptsize\vec{w}},\xi,\eta, \gamma, \rho }\text{    } \bigg( \frac{\sum_{i=1}^{l} \eta_{i} +\sum_{j=l+1}^{l+k} \min(\gamma_{j},\xi_j)}{(1/C)}  \bigg)+\|\vec{w}\| \label{semi_svm_primal}\\ 
 &\text{ s.t.:} y_i(\vec{w}^{T}\vec{\bar{h}}_{i}+\rho) \geq 1-\eta_{i} \text{, } \eta_{i}\geq0  \text{, } i=1,\ldots ,l \label{semi_svm_primal_constraint}\\
 & \hspace{.4cm}\vec{w}^{T}\vec{\bar{h}}_{j}-\rho \geq 1-\xi_{j} \text{, } \xi_{j}\geq0  \text{, } j=l+1,\ldots ,l+k  \label{semi_svm_primal_constraint2}\\
 & -\vec{w}^{T}\vec{\bar{h}}_{j}+\rho \geq 1-\gamma_{j} \text{, } \gamma_{j}\geq0  \text{, } j=l+1,\ldots ,l+k  \label{semi_svm_primal_constraint3}\\&\vec{W}^{(\cdot)T}\vec{W}^{(\cdot)}=\vec{I},  \vec{R}^{(\cdot)T}\vec{R}^{(\cdot)}=\vec{I} \text{ and } \vec{b}^{(\cdot)T}\vec{b}^{(\cdot)}=1, \label{semi_svm_primal_constraint4}
\end{align}
where $\gamma \in \mathbb{R}$ and $\eta \in \mathbb{R}$ are slack variables as $\xi$, $C$ is a trade-off parameter, $l$ and $k$ are the number of the labeled and unlabeled data instances, respectively and $y_i \in \lbrace -1,+1\rbrace$ represents the label of the $i$\textsuperscript{th} data instance.

For the quadratic programming based training method, we modify all the steps from \eqref{lagrange_svm_quad} to \eqref{derivative_W} with respect to \cref{semi_svm_primal,semi_svm_primal_constraint,semi_svm_primal_constraint2,semi_svm_primal_constraint3,semi_svm_primal_constraint4}. In a similar manner, we modify the equations from \eqref{hinge_svm} to \eqref{svm_gradient_Wgradient} according to \cref{semi_svm_primal,semi_svm_primal_constraint,semi_svm_primal_constraint2,semi_svm_primal_constraint3,semi_svm_primal_constraint4} in order to get the gradient based training method in the semi-supervised framework. For the supervised implementation, we follow the same procedure with the semi-supervised implementation for $k=0$ case.
\end{remark}
Hence, we complete the required updates for each parameter. The complete algorithm is also provided in Algorithm \ref{alg2} as a pseudocode. Moreover, we illustrate the convergence of our approximation \eqref{hinge_approx} to \eqref{hinge_svm} in Proposition \ref{proposition1}. Using Proposition \ref{proposition1}, we then demonstrate the convergence of the optimal values for our objective function \eqref{svm_unconstrained_approx} to the optimal values of the actual SVM objective function \eqref{svm_unconstrained} in Theorem \ref{theorem1}.
\begin{proposition}\label{proposition1}
As $\tau$ increases, $S_{\tau}(\beta_{\mbox{\scriptsize\vec{w}},\rho}(\vec{\bar{h}}_i)  )$ uniformly converges to $G(\beta_{\mbox{\scriptsize\vec{w}},\rho}(\vec{\bar{h}}_i)  )$. As a consequence, our approximation $F_{\tau}(\vec{w},\rho,\vec{\theta})$ converges to the SVM objective function $F(\vec{w},\rho,\vec{\theta})$, i.e., defined as
\begin{align*}
F(\vec{w},\rho,\vec{\theta})\triangleq\frac{\|\vec{w}\|^{2}}{2}+\frac{1}{n \lambda} \sum_{i=1}^{n} G(\beta_{\mbox{\scriptsize\vec{w}},\rho}(\vec{\bar{h}}_i))- \rho.
\end{align*}
\end{proposition}
\begin{proof}[Proof of Proposition \ref{proposition1}]\label{proof_proposition1}
In order to simplify our notation, for any given $\vec{w}$, $\vec{\theta}$, $\vec{X}_i$ and $\rho$, we denote $\beta_{\mbox{\scriptsize\vec{w}},\rho}(\vec{\bar{h}}_i)$ as $\Omega$. We first show that $S_{\tau}(\Omega) \geq G(\Omega  )$, $\forall \tau>0$. Since
\begin{align*}
S_{\tau}(\Omega)&= \frac{1}{\tau} \log \big(1+e^{\tau \Omega}\big)\\
&\geq\frac{1}{\tau} \log \big(e^{\tau \Omega}\big)\\
&=\Omega
\end{align*}
and $S_{\tau}(\Omega)\geq 0$, we have $S_{\tau}(\Omega)\geq G(\Omega)=\max \lbrace 0,\Omega\rbrace$. Then, for any $\Omega\geq 0$, we have
\begin{align*}
\frac{\partial S_{\tau}(\Omega)}{\partial \tau}&=\frac{-1}{\tau^2} \log \big(1+e^{\tau \Omega}\big)+\frac{1}{\tau}\frac{\Omega e^{\tau \Omega}}{1+e^{\tau \Omega}}\\
&<\frac{-1}{\tau}  \Omega+\frac{1}{\tau}\frac{\Omega e^{\tau \Omega}}{1+e^{\tau \Omega}}\\
&\leq 0
\end{align*}
and for any $\Omega<0$, we have
\begin{align*}
\frac{\partial S_{\tau}(\Omega)}{\partial \tau}&=\frac{-1}{\tau^2} \log \big(1+e^{\tau \Omega}\big)+\frac{1}{\tau}\frac{\Omega e^{\tau \Omega}}{1+e^{\tau \Omega}}\\
&<0,
\end{align*}
thus, we conclude that $S_{\tau}(\Omega)$ is a monotonically decreasing function of $\tau$. As the last step, we derive an upper bound for the difference $S_{\tau}(\Omega) - G(\Omega  )$. For $\Omega\geq 0$, the derivative of the difference is as follows
\begin{align*}
\frac{\partial (S_{\tau}(\Omega) - G(\Omega  ))}{\partial \Omega}=\frac{e^{\tau \Omega}}{1+e^{\tau \Omega}}-1<0,
\end{align*}
hence, the difference is a decreasing function of $\Omega$ for $ \Omega\geq0$. Therefore, the maximum value is $\log(2)/\tau$ and it occurs at $\Omega=0$. Similarly, for $\Omega<0$, the derivative of the difference is positive, which shows that the maximum for the difference occurs at $\Omega=0$. With this result, we obtain the following bound
\begin{align}
\label{svm_difference_bound}
\frac{\log(2)}{\tau}=\max_{\Omega}\big(S_{\tau}(\Omega) - G(\Omega  )\big).
\end{align}
Using \eqref{svm_difference_bound}, for any $\epsilon>0$, we can choose $\tau$ sufficiently large so that $S_{\tau}(\Omega) - G(\Omega  )<\epsilon$. Hence, as $\tau$ increases, $S_{\tau}(\Omega)$ uniformly converges to $G(\Omega)$. By averaging \eqref{svm_difference_bound} over all the data points and multiplying with $1/\lambda$, we obtain
\begin{align*}
\frac{\log(2)}{\lambda \tau}=\max_{\mbox{\scriptsize\vec{w}},\rho, \mbox{\scriptsize\vec{\theta}}}{\big(F_{\tau}(\vec{w},\rho,\vec{\theta}) -F(\vec{w},\rho,\vec{\theta})\big)},
\end{align*}
which proves the uniform convergence of $F_{\tau}(\cdot,\cdot,\cdot)$ to $F(\cdot,\cdot,\cdot)$.
\end{proof}
\begin{theorem}\label{theorem1}
Let $\vec{w}_{\tau}$ and $\rho_{\tau}$ be the solutions of \eqref{svm_unconstrained_approx} for any fixed $\vec{\theta}$. Then, $\vec{w}_{\tau}$ and $\rho_{\tau}$ are unique and $F_{\tau}(\vec{w}_{\tau},\rho_{\tau},\vec{\theta})$ converges to the minimum of $F(\vec{w},\rho,\vec{\theta})$.
\end{theorem}
\begin{proof}[Proof of Theorem \ref{theorem1}]\label{proof_theorem1}
We have the following Hessian matrix of $F_{\tau}(\vec{w},\rho,\vec{\theta})$ with respect to $\vec{w}$
\begin{align*}
\nabla_{ \mbox{\scriptsize \vec{w}}}^{2}F_{\tau}(\vec{w},\rho,\vec{\theta})=\vec{I}+\frac{\tau}{n \lambda} \sum_{i=1}^{n} \frac{e^{\tau \beta_{\mbox{\scriptsize\vec{w}},\rho}(\vec{\bar{\mbox{\scriptsize h}}}_i) }}{\big(1+e^{\tau \beta_{\mbox{\scriptsize\vec{w}},\rho}(\vec{\bar{\mbox{\scriptsize h}}}_i) } \big)^2}\vec{\bar{h}}_i \vec{\bar{h}}_i^T,
\end{align*}
which satisfies $\vec{v}^T\nabla_{ \mbox{\scriptsize \vec{w}}}^{2}F_{\tau}(\vec{w},\rho,\vec{\theta})\vec{v}>0 $ for any nonzero column vector $\vec{v}$. Hence, the Hessian matrix is positive definite, which shows that $F_{\tau}(\vec{w},\rho,\vec{\theta})$ is strictly convex function of $\vec{w}$. Consequently, the solution $\vec{w}_{\tau}$ is both global and unique given any $\rho$ and $\vec{\theta}$. Additionally, we have the following second order derivative for $\rho$
\begin{align*}
\frac{\partial^2 F_{\tau}(\vec{w},\rho,\vec{\theta})}{\partial \rho^2}=\frac{\tau}{n \lambda} \sum_{i=1}^{n} \frac{e^{\tau \beta_{\mbox{\scriptsize\vec{w}},\rho}(\vec{\bar{\mbox{\scriptsize h}}}_i) }}{\big(1+e^{\tau \beta_{\mbox{\scriptsize\vec{w}},\rho}(\vec{\bar{\mbox{\scriptsize h}}}_i) } \big)^2} > 0,
\end{align*}
which implies that $F_{\tau}(\vec{w},\rho,\vec{\theta})$ is strictly convex function of $\rho$. As a result, the solution $\rho_{\tau}$ is both global and unique for any given $\vec{w}$ and $\vec{\theta}$.

Let $\vec{w}^{*}$ and $\rho^{*}$ be the solutions of \eqref{svm_unconstrained} for any fixed $\vec{\theta}$. From the proof of Proposition \ref{proposition1}, we have 
\begin{align}
\label{theorem1_eq1}
F_{\tau}(\vec{w}^*,\rho^*, \vec{\theta})\geq F_{\tau}(\vec{w}_{\tau},\rho_{\tau}, \vec{\theta}) &\geq F(\vec{w}_{\tau},\rho_{\tau}, \vec{\theta}) \nonumber \\& \geq F(\vec{w}^*,\rho^*, \vec{\theta}).
\end{align}
Using the convergence result in Proposition \ref{proposition1} and \eqref{theorem1_eq1}, we have
\begin{align*}
&\lim_{\tau \rightarrow \infty} F_{\tau}(\vec{w}_{\tau},\rho_{\tau}, \vec{\theta})\leq \lim_{\tau \rightarrow \infty} F_{\tau}(\vec{w}^*,\rho^*, \vec{\theta})=F(\vec{w}^*,\rho^*, \vec{\theta})\\
&\lim_{\tau \rightarrow \infty} F_{\tau}(\vec{w}_{\tau},\rho_{\tau}, \vec{\theta}) \geq F(\vec{w}^*,\rho^*, \vec{\theta}),
\end{align*}
which proves the following equality
\begin{align*}
\lim_{\tau \rightarrow \infty} F_{\tau}(\vec{w}_{\tau},\rho_{\tau}, \vec{\theta})=F(\vec{w}^*,\rho^*, \vec{\theta}).
\end{align*}
\end{proof}

\subsection{Anomaly Detection with the SVDD algorithm}\label{sec:sub_svdd}
In this subsection, we introduce an anomaly detection algorithm based on the SVDD formulation and provide the joint updates in order to learn both the LSTM and SVDD parameters. However, since the generic formulation is the same with the OC-SVM case, we only provide the required and distinct updates for the parameters and proof for the convergence of the approximated SVDD formulation to the actual one.

In the SVDD algorithm, we aim to find a hypersphere that encloses the normal data while leaving the anomalous data outside the hypersphere \cite{svdd}. For the sequence $\lbrace \vec{\bar{h}}_i \rbrace _{i=1}^{n}$, we have the following SVDD optimization problem \cite{svdd}
\begin{align}
\label{svdd_primal}
&\min_{\mbox{\scriptsize \vec{\theta}} \in \mathbb{R}^{n_{\theta}},\mbox{\scriptsize \vec{\tilde{c}}} \in \mathbb{R}^m, \xi \in \mathbb{R},R\in\mathbb{R}} \text{  }  R^2+ \frac{1}{n \lambda}\sum_{i=1}^n \xi_i\\
 &\text{ subject to: } \| \vec{\bar{h}}_i-\vec{\tilde{c}} \|^2 - R^2 \leq \xi_i\text{, } \xi_{i}\geq0 , \forall i \label{svdd_primal_constraint}\\
 &\vec{W}^{(\cdot)T}\vec{W}^{(\cdot)}=\vec{I},  \vec{R}^{(\cdot)T}\vec{R}^{(\cdot)}=\vec{I} \text{ and } \vec{b}^{(\cdot)T}\vec{b}^{(\cdot)}=1, \label{svdd_primal_constraint2}
\end{align}
where $\lambda>0$ is a trade-off parameter between $R^2$ and the total misclassification error, $R$ is the radius of the hypersphere and $\vec{\tilde{c}}$ is the center of the hypersphere. Additionally, $\vec{\theta}$ and $\xi$ represent the LSTM parameters and the slack variable respectively as in the OC-SVM case. After solving the constrained optimization problem in \eqref{svdd_primal}, \eqref{svdd_primal_constraint} and \eqref{svdd_primal_constraint2}, we detect anomalies using the following scoring function
\begin{align}
\label{decision_svdd}
l(\vec{X}_i)=\sgn(R^2 -\| \vec{\bar{h}}_i-\vec{\tilde{c}} \|^2).
\end{align}

\begin{algorithm}
\caption{Gradient Based Training for the Anomaly Detection Algorithm Based on OC-SVM}
\label{alg2}
\begin{algorithmic}[1]
\State{Initialize the LSTM parameters as $\vec{\theta}_0$ and the OC-SVM parameters as $\vec{w}_0$ and $\rho_0$}
\State{Determine a threshold $\epsilon$ as convergence criterion}
\State{$k=-1$}
\Do
\State{$k=k+1$}
\State{Using $\vec{\theta}_k$, obtain $\lbrace \vec{\bar{h}}\rbrace_{i=1}^n$ according to Fig. \ref{pooling}}
\State{Obtain $\vec{w}_{k+1}$, $\rho_{k+1}$ and $\vec{\theta}_{k+1}$ using \eqref{update_w}, \eqref{update_rho}, \eqref{svm_gradient_Wupdate} and Remark \ref{remark_svm_gradient} }
\doWhile{$\big( F_{\tau}(\vec{w}_{k+1},\rho_{k+1}, \vec{\theta}_{k+1}) -F_{\tau}(\vec{w}_{k},\rho_{k}, \vec{\theta}_{k})\big)^2>\epsilon$}
\State{Detect anomalies using \eqref{decision_svm} evaluated at $\vec{w}_k$, $\rho_k$ and $\vec{\theta}_k$ }
\end{algorithmic}
\end{algorithm}
\subsubsection{Quadratic Programming Based Training Algorithm}
In this subsection, we introduce a training algorithm based on quadratic programming for \eqref{svdd_primal}, \eqref{svdd_primal_constraint} and \eqref{svdd_primal_constraint2}. As in the OC-SVM case, we first assume that the LSTM parameters are fixed and then perform optimization over the SVDD parameters based on the fixed LSTM parameters. For \eqref{svdd_primal} and \eqref{svdd_primal_constraint}, we have the following Lagrangian
\begin{align}
L(\vec{\tilde{c}},\xi,R,\nu,\alpha)=&R^2+ \frac{1}{n\lambda}\sum_{i=1}^n \xi_i-\sum_{i=1}^{n} \nu_{i}\xi_{i} \nonumber \\
&-\sum_{i=1}^{n} \alpha_{i}(\xi_{i}-\| \vec{\bar{h}}_i-\vec{\tilde{c}} \|^2 + R^2) \label{lagrange_svdd_quad},
\end{align}
where $\nu_{i}$, $\alpha_{i} \geq 0$ are the Lagrange multipliers. Taking derivative of \eqref{lagrange_svdd_quad} with respect to $\vec{\tilde{c}}$, $\xi$ and $R$ and then setting the derivatives to zero yields
\begin{align}
&\vec{\tilde{c}}=\sum_{i=1}^{n}\alpha_{i}\vec{\bar{h}}_{i} \label{clagrange_svdd_quad}\\
&\sum_{i=1}^{n}\alpha_{i}=1 \text{ and } \alpha_{i}=1/(n\lambda)-\nu_{i} \text{, } \forall i. \label{alphalagrange_quad}
\end{align}
Putting \eqref{clagrange_svdd_quad} and \eqref{alphalagrange_quad} into \eqref{lagrange_svdd_quad}, we obtain a dual form for \eqref{svdd_primal} and \eqref{svdd_primal_constraint} as follows
\begin{align}
\label{svdd_dual}
&\min_{\mbox{\scriptsize \vec{\theta}}\in \mathbb{R}^{n_{\theta}}, \mbox{\scriptsize\vec{\alpha}} \in \mathbb{R}^n} \sum_{i=1}^{n} \sum_{j=1}^{n} \alpha_{i}\alpha_{j} \vec{\bar{h}}_i^T \vec{\bar{h}}_j - \sum_{i=1}^n \alpha_i \vec{\bar{h}}_i^T \vec{\bar{h}}_i \\
\label{svdd_dual_constraint}
&\text{subject to: } \sum_{i=1}^n \alpha_i =1 \text{ and } 0 \leq \alpha_i \leq 1/(n\lambda), \forall i\\
 &\vec{W}^{(\cdot)T}\vec{W}^{(\cdot)}=\vec{I},  \vec{R}^{(\cdot)T}\vec{R}^{(\cdot)}=\vec{I} \text{ and } \vec{b}^{(\cdot)T}\vec{b}^{(\cdot)}=1. \label{svdd_dual_constraint2}
\end{align}
Using \eqref{clagrange_svdd_quad}, we modify \eqref{decision_svdd} as
\begin{align}
\label{svdd_dual_decision}
l(\vec{X}_i)=\sgn \bigg( R^2-\sum_{k=1}^{n}& \sum_{j=1}^{n} \alpha_{k}\alpha_{j}  \vec{\bar{h}}_k^T \vec{\bar{h}}_j \nonumber\\& +2 \sum_{j=1}^n \alpha_j  \vec{\bar{h}}_j^T \vec{\bar{h}}_i-\vec{\bar{h}}_i^T \vec{\bar{h}}_i \bigg).
\end{align}
In order to solve the constrained optimization problem in \eqref{svdd_dual}, \eqref{svdd_dual_constraint} and \eqref{svdd_dual_constraint2}, we employ the same approach as in the OC-SVM case. We first fix a certain set of the LSTM parameters $\vec{\theta}$. Based on these parameters, we find the optimal $\vec{\alpha}$ using the SMO algorithm. After that, we fix $\vec{\alpha}$ to update $\vec{\theta}$ using the algorithm for optimization with orthogonality constraints. We repeat these procedures until we reach convergence. Finally, we evaluate \eqref{svdd_dual_decision} based on the converged parameters.

\begin{remark}
\label{svdd_remark_alpha}
In the SVDD case, we apply the SMO algorithm using the same procedures with the OC-SVM case. In particular, we first choose two parameters, e.g., $\alpha_{1}$ and $\alpha_{2}$, to minimize and fix the other parameters. Due to \eqref{svdd_dual_constraint}, the chosen parameters must obey \eqref{smo1}. Hence, we have the following update rule for $\alpha_2$ at the $k$\textsuperscript{th} iteration
\begin{align*}
\alpha_{k+1,2}=\frac{2(1-S)(K_{11}-K_{12})+K_{22}-K_{11}+M_1 - M_2}{2(K_{11}+K_{22}-2K_{12})},
\end{align*}
where $S=\sum_{j=3}^n \alpha_{k,j}$ and the other definitions are the same with the OC-SVM case. We then obtain $\alpha_{k+1,1}$ using \eqref{smo1}. By this, we obtain the updated values $\alpha_{k+1,2}$ and $\alpha_{k+1,1}$. For the remaining parameters, we repeat this procedure until reaching convergence.
\end{remark}
\begin{remark}
\label{svdd_dual_theta_remark}
For the SVDD case, we update $\vec{W}^{(\cdot)}$ at the $k$\textsuperscript{th} iteration as in \eqref{svm_dual_Wupdate}. However, instead of \eqref{derivative_W}, we have the following definition for $\vec{G}$
\begin{align*}
\vec{G}_{ij} = \frac{\partial \pi (\vec{\theta},\vec{\alpha}_{k+1})}{\partial \vec{W}^{(\cdot)}_{ij}},
\end{align*}
where
\begin{align*}
\pi(\vec{\theta},\vec{\alpha}_{k+1})\triangleq \sum_{i=1}^{n} \sum_{j=1}^{n} \alpha_{k+1,i}\alpha_{k+1,j} \vec{\bar{h}}_i^T \vec{\bar{h}}_j - \sum_{i=1}^n \alpha_{k+1,i} \vec{\bar{h}}_i^T \vec{\bar{h}}_i 
\end{align*}
at the $k$\textsuperscript{th} iteration. For the remaining parameters, we follow the procedure in Remark \ref{remark2}.
\end{remark}
Hence, we obtain a quadratic programming based training algorithm for our LSTM based anomaly detector, which is also described in Algorithm \ref{alg3} as a pseudocode.
\begin{algorithm}
\caption{Quadratic Programming Based Training for the Anomaly Detection Algorithm Based on SVDD}
\label{alg3}
\begin{algorithmic}[1]
\State{Initialize the LSTM parameters as $\vec{\theta}_0$ and the dual SVDD parameters as $\vec{\alpha}_0$}
\State{Determine a threshold $\epsilon$ as convergence criterion}
\State{$k=-1$}
\Do
\State{$k=k+1$}
\State{Using $\vec{\theta}_k$, obtain $\lbrace \vec{\bar{h}}\rbrace_{i=1}^n$ according to Fig. \ref{pooling}}
\State{Find optimal $\vec{\alpha}_{k+1}$ for $\lbrace \vec{\bar{h}}\rbrace_{i=1}^n$ using the procedure in Remark \ref{svdd_remark_alpha}}
\State{Based on $\vec{\alpha}_{k+1}$, obtain $\vec{\theta}_{k+1} $ using Remark \ref{svdd_dual_theta_remark}}
\doWhile{$\big( \pi(\vec{\theta}_{k+1},\vec{\alpha}_{k+1}) -\pi(\vec{\theta}_k,\vec{\alpha}_{k})\big)^2 >\epsilon$}
\State{Detect anomalies using \eqref{svdd_dual_decision} evaluated at $\vec{\theta}_k$ and $\vec{\alpha}_k$ }
\end{algorithmic}
\end{algorithm}
\subsubsection{Gradient Based Training Algorithm}
In this subsection, we introduce a training algorithm based on only the first order gradients for \eqref{svdd_primal}, \eqref{svdd_primal_constraint} and \eqref{svdd_primal_constraint2}. We again use the $G(\cdot)$ function in \eqref{hinge_svm} in order to eliminate the constraint in \eqref{svdd_primal_constraint} as follows
\begin{align}
\label{svdd_unconstrained}
&\min_{\mbox{\scriptsize \vec{\theta}} \in \mathbb{R}^{n_{\theta}},\mbox{\scriptsize \vec{\tilde{c}}} \in \mathbb{R}^m,R\in\mathbb{R}} \text{  }  R^2+ \frac{1}{n\lambda}\sum_{i=1}^nG(\Psi_{R,\mbox{\scriptsize \vec{\tilde{c}}}}(\vec{\bar{h}}_i))\\
 \label{svdd_unconstrained_constrained}
& \text{s.t.:}\vec{W}^{(\cdot)T}\vec{W}^{(\cdot)}=\vec{I},  \vec{R}^{(\cdot)T}\vec{R}^{(\cdot)}=\vec{I} \text{ and } \vec{b}^{(\cdot)T}\vec{b}^{(\cdot)}=1,
\end{align}
where 
\begin{align*}
\Psi_{R,\mbox{\scriptsize \vec{\tilde{c}}}}(\vec{\bar{h}}_i)\triangleq \| \vec{\bar{h}}_i-\vec{\tilde{c}} \|^2 - R^2.
\end{align*}
Since the gradient based methods cannot optimize \eqref{svdd_unconstrained} due to the nondifferentiable function $G(\cdot)$, we employ $S_{\tau}(\cdot)$ instead of $G(\cdot)$ and modify \eqref{svdd_unconstrained} as 
\begin{align}
\label{svdd_unconstrained_approx}
&\min_{\mbox{\scriptsize \vec{\theta}} \in \mathbb{R}^{n_{\theta}},\mbox{\scriptsize \vec{\tilde{c}}} \in \mathbb{R}^m,R\in\mathbb{R}} \text{} F_{\tau}(\vec{\tilde{c}},R,\vec{\theta})=  R^2+\frac{1}{n\lambda} \sum_{i=1}^n S_{\tau}(\Psi_{R,\mbox{\scriptsize \vec{\tilde{c}}}}(\vec{\bar{h}}_i))\\
 \label{svdd_unconstrained_approx_constrained}
& \text{s.t.:}\vec{W}^{(\cdot)T}\vec{W}^{(\cdot)}=\vec{I},  \vec{R}^{(\cdot)T}\vec{R}^{(\cdot)}=\vec{I} \text{ and } \vec{b}^{(\cdot)T}\vec{b}^{(\cdot)}=1,
\end{align}
where $F_{\tau}(\cdot,\cdot,\cdot)$ is the objective function of \eqref{svdd_unconstrained_approx}. To obtain the optimal values for \eqref{svdd_unconstrained_approx} and \eqref{svdd_unconstrained_approx_constrained}, we update $\vec{\tilde{c}}$, $R$ and $\vec{\theta}$ till we reach either a local or a global optimum. For the updates of $\vec{\tilde{c}}$ and $R$, we employ the SGD algorithm, where we use the following gradient calculations. We first compute the gradient of $\vec{\tilde{c}}$ as
\begin{align}
\label{gradient_c}
\nabla_{\mbox{\scriptsize \vec{\tilde{c}}}}F_{\tau}(\vec{\tilde{c}},R,\vec{\theta})=\frac{1}{n\lambda}\sum_{i=1}^n
\frac{2(\vec{\tilde{c}}-\vec{\bar{h}}_i) e^{\tau \Psi_{\mbox{\scriptsize\vec{\tilde{c}}},R}(\vec{\bar{\mbox{\scriptsize h}}}_i) }}{1+e^{\tau \Psi_{\mbox{\scriptsize\vec{\tilde{c}}},R}(\vec{\bar{\mbox{\scriptsize h}}}_i) } }.
\end{align}
Using \eqref{gradient_c}, we have the following update 
\begin{align}
\label{update_c}
\vec{\tilde{c}}_{k+1}=\vec{\tilde{c}}_{k}-\mu\nabla_{\mbox{\scriptsize \vec{\tilde{c}}}}F_{\tau}(\vec{\tilde{c}},R,\vec{\theta})\Bigr|_{\scriptsize \substack{  \vec{\tilde{c}}=\vec{\tilde{c}}_{k} \\ R^2=R_{k}^2
\\  \vec{\theta}=\vec{\theta}_{k}}},
\end{align}
where the subscript $k$ represents the iteration number. Likewise, we compute the derivative of the objective function with respect to $R^2$ as
\begin{align}
\label{gradient_r}
\frac{\partial F_{\tau}(\vec{\tilde{c}},R,\vec{\theta})}{\partial R^2}=1+\frac{1}{n\lambda}\sum_{i=1}^n
\frac{- e^{\tau \Psi_{\mbox{\scriptsize\vec{\tilde{c}}},R}(\vec{\bar{\mbox{\scriptsize h}}}_i) }}{1+e^{\tau \Psi_{\mbox{\scriptsize\vec{\tilde{c}}},R}(\vec{\bar{\mbox{\scriptsize h}}}_i) } }.
\end{align}
With \eqref{gradient_r}, we update $R^2$ as
\begin{align}
\label{update_r}
R^2_{k+1}=R^2_{k}-\mu\frac{\partial F_{\tau}(\vec{\tilde{c}},R,\vec{\theta})}{\partial R^2}\Bigr|_{\scriptsize \substack{  \vec{\tilde{c}}=\vec{\tilde{c}}_{k} \\ R^2=R_{k}^2
\\  \vec{\theta}=\vec{\theta}_{k}}}.
\end{align}
For $\vec{\theta}$, the gradient calculation is as follows
\begin{align}
\label{svdd_gradient_W}
\frac{\partial F_{\tau}(\vec{\tilde{c}},R,\vec{\theta})}{\partial \vec{W}^{(\cdot)}_{ij}}=\sum_{i=1}^n
\frac{2(\partial \vec{\bar{h}}_i/\partial \vec{W}^{(\cdot)}_{ij})^T(\vec{\bar{h}}_i-\vec{\tilde{c}}) e^{\tau \Psi_{\mbox{\scriptsize\vec{\tilde{c}}},R}(\vec{\bar{\mbox{\scriptsize h}}}_i) }}{n\lambda(1+e^{\tau \Psi_{\mbox{\scriptsize\vec{\tilde{c}}},R}(\vec{\bar{\mbox{\scriptsize h}}}_i) }) }.
\end{align}
Using \eqref{svdd_gradient_W}, we have the following update
\begin{align}
\label{svdd_W_update}
\vec{W}_{k+1}^{(\cdot)}= \bigg( \vec{I}+\frac{\mu}{2} \vec{B}_{k} \bigg)^{-1} \bigg( \vec{I}-\frac{\mu}{2} \vec{B}_{k} \bigg)\vec{W}_{k}^{(\cdot)},
\end{align}
where $\vec{B}_k=\vec{M}_k (\vec{W}^{(\cdot)}_k)^{T}-\vec{W}^{(\cdot)}_k\vec{M}_k^T$ and 
\begin{align}
\label{svdd_W_derivative}
\vec{M}_{ij} \triangleq\frac{ \partial F_{\tau}(\vec{\tilde{c}},R,\vec{\theta})}{\partial \vec{W}^{(\cdot)}_{ij}}.
\end{align}
\begin{remark}
\label{remark_svdd_gradient}
For $\vec{R}^{(\cdot)}$ and $\vec{b}^{(\cdot)}$, we first compute the gradient of the objective function with respect to the chosen parameter as in \eqref{svdd_W_derivative}. We then obtain $\vec{B}_k$ according to the chosen parameter. Using $\vec{B}_k$, we update the chosen parameter as in \eqref{svdd_W_update}.
\end{remark}
\begin{remark}
\label{svd_semi}
In the semi-supervised framework, we have the following optimization problem for our SVDD based algorithms \cite{supervised_anomaly}
\begin{align}
&\min_{\mbox{\scriptsize\vec{\theta}} ,\mbox{\scriptsize\vec{\tilde{c}}},R,\xi, \gamma, \eta }\text{    } R^2-C_1 \gamma+C_2 \sum_{i=1}^l \xi_i+C_3 \sum_{j=l+1}^{l+k} \eta_j \label{semi_svdd_primal}\\ 
 &\text{ s.t.:}  \| \vec{\bar{h}}_i-\vec{\tilde{c}} \|^2 - R^2 \leq \xi_i \text{, } \xi_{i}\geq0\text{, } \forall_{i=1}^{l} \label{semi_svdd_primal_constraint}\\
 & \hspace{0.3cm}y_j(\| \vec{\bar{h}}_j-\vec{\tilde{c}} \|^2 - R^2) \leq -\gamma+\eta_j\text{, }  \eta_{j}\geq0 \text{, } \forall_{j=l+1}^{l+k}  \label{semi_svdd_primal_constraint2}\\
 &\vec{W}^{(\cdot)T}\vec{W}^{(\cdot)}=\vec{I},  \vec{R}^{(\cdot)T}\vec{R}^{(\cdot)}=\vec{I} \text{ and } \vec{b}^{(\cdot)T}\vec{b}^{(\cdot)}=1, \label{semi_svdd_primal_constraint3}
\end{align}
where $\eta \in \mathbb{R}$ is a slack variable as $\xi$, $\gamma \in \mathbb{R}$ is the margin of the labeled data instances, $C_1$, $C_2$ and $C_3$ are trade-off parameters, $k$ and $l$ are the number of the labeled and unlabeled data instances, respectively and $y_j \in \lbrace -1,+1\rbrace$ represents the label of the $j$\textsuperscript{th} data instance.

For the quadratic programming based training method, we modify all the steps from \eqref{lagrange_svdd_quad} to \eqref{svdd_dual_decision}, Remark \ref{svdd_remark_alpha} and Remark \ref{svdd_dual_theta_remark} with respect to \cref{semi_svdd_primal,semi_svdd_primal_constraint,semi_svdd_primal_constraint2,semi_svdd_primal_constraint3}. In a similar manner, we modify the equations from \eqref{svdd_unconstrained} to \eqref{svdd_W_derivative} according to \cref{semi_svdd_primal,semi_svdd_primal_constraint,semi_svdd_primal_constraint2,semi_svdd_primal_constraint3} in order to obtain the gradient based training method in the semi-supervised framework. For the supervised implementation, we follow the same procedure with the semi-supervised implementation for $l=0$ case.
\end{remark}
The complete algorithm is provided in Algorithm \ref{alg4}. In the following, we provide the convergence proof as in the OC-SVM case.
\begin{theorem}\label{theorem2}
Let $\vec{\tilde{c}}_{\tau}$ and $R_{\tau}^2$ be the solutions of \eqref{svdd_unconstrained_approx} for any fixed $\vec{\theta}$. Then, $\vec{\tilde{c}}_{\tau}$ and $R_{\tau}^2$ are unique and $F_{\tau}(\vec{\tilde{c}}_{\tau},R_{\tau},\vec{\theta})$ converges to the minimum of $F(\vec{\tilde{c}},R,\vec{\theta})$, i.e., defined as
\begin{align*}
F(\vec{\tilde{c}},R,\vec{\theta}) \triangleq R^2+ \frac{1}{n\lambda}\sum_{i=1}^nG(\Psi_{R,\mbox{\scriptsize \vec{\tilde{c}}}}(\vec{\bar{h}}_i)).
\end{align*}
\end{theorem}
\begin{proof}[Proof of Theorem \ref{theorem2}]\label{proof_theorem2}
We have the following Hessian matrix of $F_{\tau}(\vec{\vec{\tilde{c}}},R,\vec{\theta})$ with respect to $\vec{\tilde{c}}$
\begin{align*}
\nabla_{ \mbox{\scriptsize \vec{\tilde{c}}}}^{2}F_{\tau}(\vec{\tilde{c}},R,\vec{\theta})=\sum_{i=1}^{n} \frac{2\vec{I}(\Omega_i+\Omega_i^2)+4\tau \Omega_i (\vec{\tilde{c}}-\vec{\bar{h}}_i)(\vec{\tilde{c}}-\vec{\bar{h}}_i)^T}{n\lambda\big(1+\Omega_i \big)^2},
\end{align*}
where $\Omega_i = e^{\tau\Psi_{\mbox{\scriptsize\vec{\tilde{c}}},R}(\vec{\bar{\mbox{\scriptsize h}}}_i)  }$, which implies $\vec{v}^T\nabla_{ \mbox{\scriptsize \vec{\tilde{c}}}}^{2}F_{\tau}(\vec{\tilde{c}},R,\vec{\theta})\vec{v}>0 $ for any nonzero column vector $\vec{v}$. Thus, the Hessian matrix is positive definite, which shows that $F_{\tau}(\vec{w},\rho,\vec{\theta})$ is strictly convex function of $\vec{\tilde{c}}$. As a result, the solution $\vec{\tilde{c}}_{\tau}$ is both global and unique given any $R$ and $\vec{\theta}$. In addition to this, we have the following second order derivative for $R^2$
\begin{align*}
\frac{\partial^2 F_{\tau}(\vec{\tilde{c}},R,\vec{\theta})}{\partial (R^2)^2}=\frac{\tau}{n \lambda} \sum_{i=1}^{n} \frac{e^{\tau \Psi_{\mbox{\scriptsize\vec{\tilde{c}}},R}(\vec{\bar{\mbox{\scriptsize h}}}_i)  }}{\big(1+e^{\tau \Psi_{\mbox{\scriptsize\vec{\tilde{c}}},R}(\vec{\bar{\mbox{\scriptsize h}}}_i) } \big)^2} > 0,
\end{align*}
which implies that $F_{\tau}(\vec{\tilde{c}},R,\vec{\theta})$ is strictly convex function of $R^2$. Therefore, the solution $R_{\tau}^2$ is both global and unique for any given $\vec{\tilde{c}}$ and $\vec{\theta}$.

The convergence proof directly follows the proof of Theorem \ref{theorem1}.
\end{proof}
\begin{algorithm}
\caption{Gradient Based Training for the Anomaly Detection Algorithm Based on SVDD}
\label{alg4}
\begin{algorithmic}[1]
\State{Initialize the LSTM parameters as $\vec{\theta}_0$ and the SVDD parameters as $\vec{\tilde{c}}_0$ and $R^2_0$}
\State{Determine a threshold $\epsilon$ as convergence criterion}
\State{$k=-1$}
\Do
\State{$k=k+1$}
\State{Using $\vec{\theta}_k$, obtain $\lbrace \vec{\bar{h}}\rbrace_{i=1}^n$ according to Fig. \ref{pooling}}
\State{Obtain $\vec{\tilde{c}}_{k+1}$, $R^2_{k+1}$ and $\vec{\theta}_{k+1}$ using \eqref{update_c}, \eqref{update_r}, \eqref{svdd_W_update} and Remark \ref{remark_svdd_gradient} }
\doWhile{$\big( F_{\tau}(\vec{\tilde{c}}_{k+1},R_{k+1}, \vec{\theta}_{k+1}) -F_{\tau}(\vec{\tilde{c}}_{k},R_{k}, \vec{\theta}_{k})\big)^2>\epsilon$}
\State{Detect anomalies using \eqref{decision_svdd} evaluated at $\vec{\tilde{c}}_k$, $R^2_k$ and $\vec{\theta}_k$ }
\end{algorithmic}
\end{algorithm}
\section{Simulations}\label{sec:simulations}
In this section, we demonstrate the performances of the algorithms on several different datasets. We first evaluate the performances on a dataset that contains variable length data sequences, i.e., the digit dataset \cite{uci}. We then compare the anomaly detection performances on several different benchmark real datasets such as the occupancy \cite{occupancy}, Hong Kong Exchange (HKE) rate \cite{hke}, http \cite{http} and Alcoa stock price \cite{alcoa} datasets. While performing experiments on real benchmark datasets, we also include the GRU based algorithms in order to compare their performances with the LSTM based ones. Note that since the introduced algorithms have bounded functions, e.g., the sigmoid function in the LSTM architecture, for all the experiments in this section, we normalize each dimension of the datasets into $[-1, 1]$.

Throughout this section, we denote the LSTM based OC-SVM anomaly detectors that are trained with the gradient and quadratic programming based algorithms as ``LSTM-GSVM" and ``LSTM-QPSVM", respectively. In a similar manner, we use ``LSTM-GSVDD" and ``LSTM-QPSVDD" for the SVDD based anomaly detectors. Moreover, for the labels of the GRU based algorithms, we replace the LSTM prefix with GRU. 

\subsection{Anomaly Detection for Variable Length Data Sequences}
In this section, we evaluate the performances of the introduced anomaly detectors on the digit dataset \cite{uci}. In this dataset, we have the pixel samples of digits, which were written on a tablet by several different authors \cite{uci}. Since the speed of writing varies from person to person, the number of samples for a certain digit might significantly differ. The introduced algorithms are able to process such kind of sequences thanks to their generic structure in Fig. \ref{pooling}. However, the conventional OC-SVM and SVDD algorithms cannot directly process these sequences \cite{svm1,svdd}. For these algorithms, we take the mean of each sequence to obtain a fixed length vector sequence, i.e., two dimensional in this case (two coordinates of a pixel). In order to evaluate the performances, we first choose a digit as normal and another digit as anomaly. We emphasize that randomly choose digits for illustration and we obtain similar performances for the other digits. We then divide the samples of these digits into training and test parts, where we allocate $60\%$ of the samples for the training part and $40\%$ for the test part. In both the training and test parts, we select the samples so that $10\%$ of the samples are anomalies. Then, using the training part, we optimize the parameters of each algorithm using two fold cross validation, where we also select certain crucial parameter, e.g., $\mu$. This procedure results in $\mu=0.05, 0.001,0.05$ and $0.01$ for LSTM-GSVM, LSTM-QPSVM, LSTM-GSVDD and LSTM-QPSVDD, respectively. Furthermore, we select the output dimension of the LSTM architecture as $m=2$ and the regularization parameter as $\lambda=0.5$ for all the algorithms. For the implementation of the conventional OC-SVM and SVDD algorithms, we use the libsvm library and their parameters are selected in a similar manner via built in optimization tools of libsvm \cite{libsvm}.

	 \begin{figure*}[h]
	\centering
	\captionsetup[subfigure]{oneside,margin={1cm,0cm}}
	\begin{subfigure}[t]{0.45\textwidth}
		\centering
		\includegraphics[width=1.11\textwidth, height=0.78\textwidth]{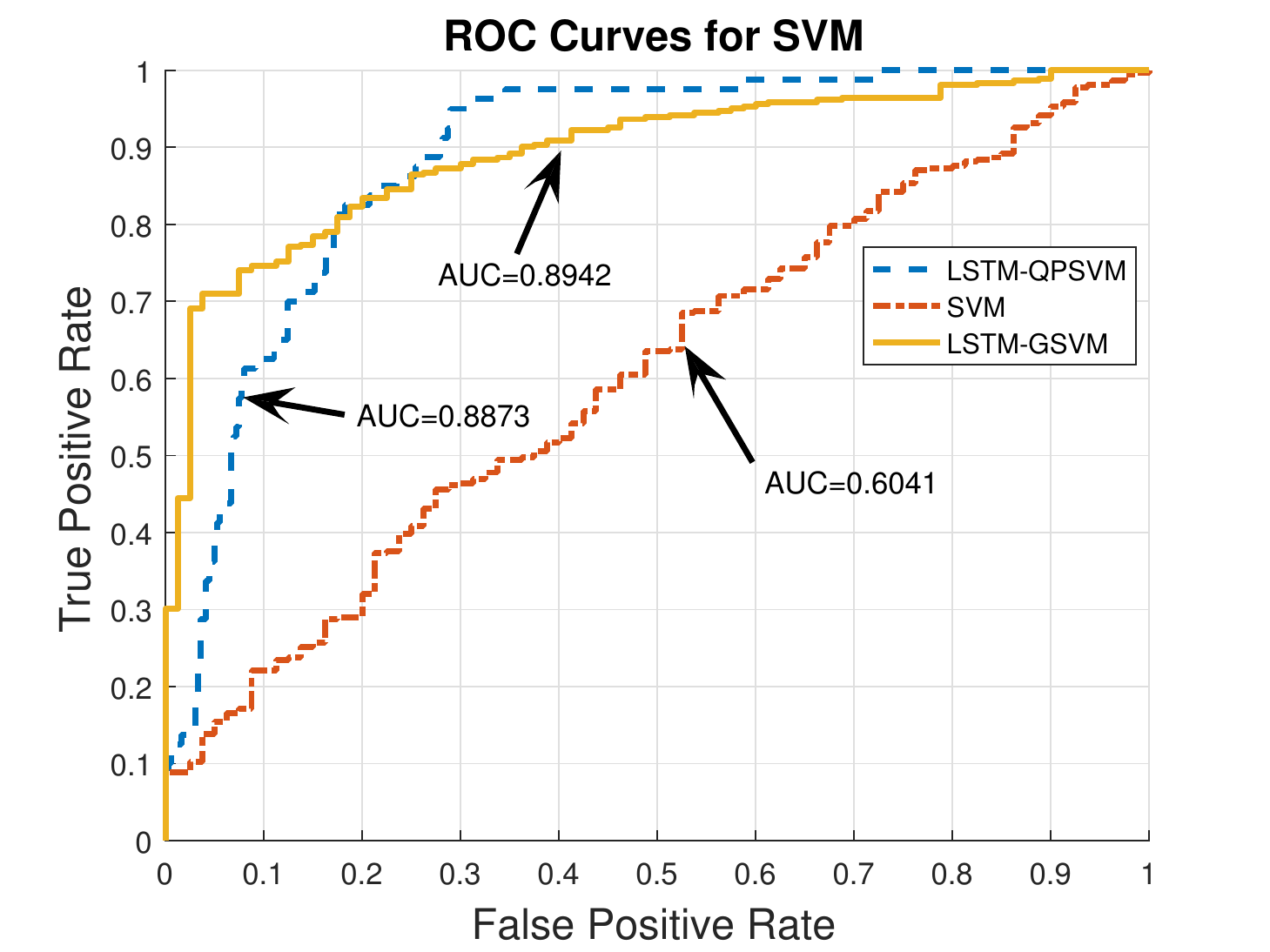}
	\caption{\centering\medskip} \label{svm0-9}
	\end{subfigure}\hspace*{\fill}
		\begin{subfigure}[t]{0.45\textwidth}
		\centering
		\includegraphics[width=1.11\textwidth, height=0.78\textwidth]{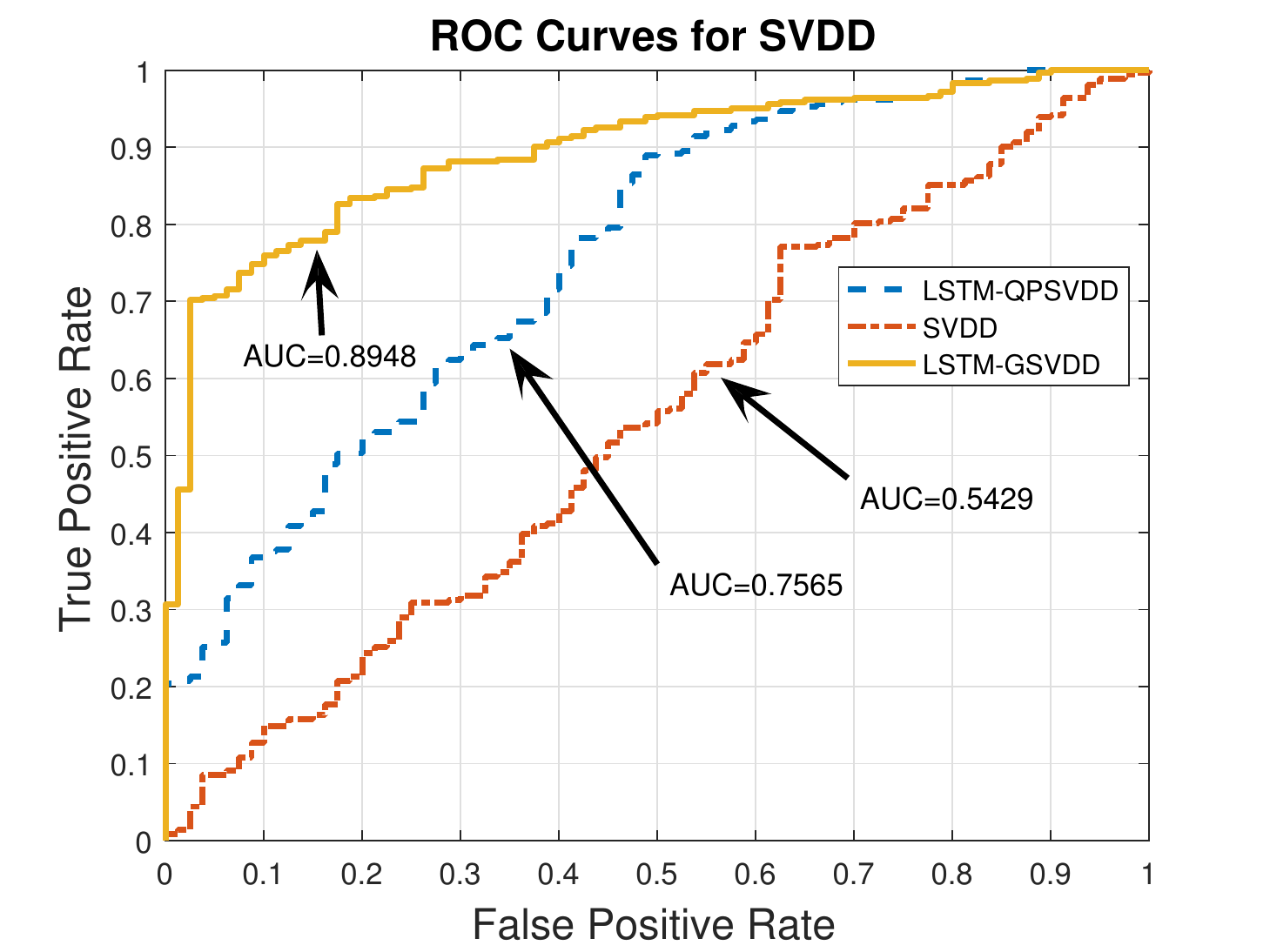}
		\caption{\centering\medskip} \label{svdd0-9}
	\end{subfigure} \hspace*{\fill}
	\medskip
	\caption{The ROC curves of the algorithms for the digit dataset, where we consider digit ``0" as normal and digit ``9" as anomaly (a) for the SVM based algorithms and (b) for the SVDD based algorithms.}\label{figs1}
	\end{figure*}
	
		 \begin{figure*}[h]
	\centering
	\captionsetup[subfigure]{oneside,margin={1cm,0cm}}
		\begin{subfigure}[t]{0.45\textwidth}
		\centering
		\includegraphics[width=1.11\textwidth, height=0.78\textwidth]{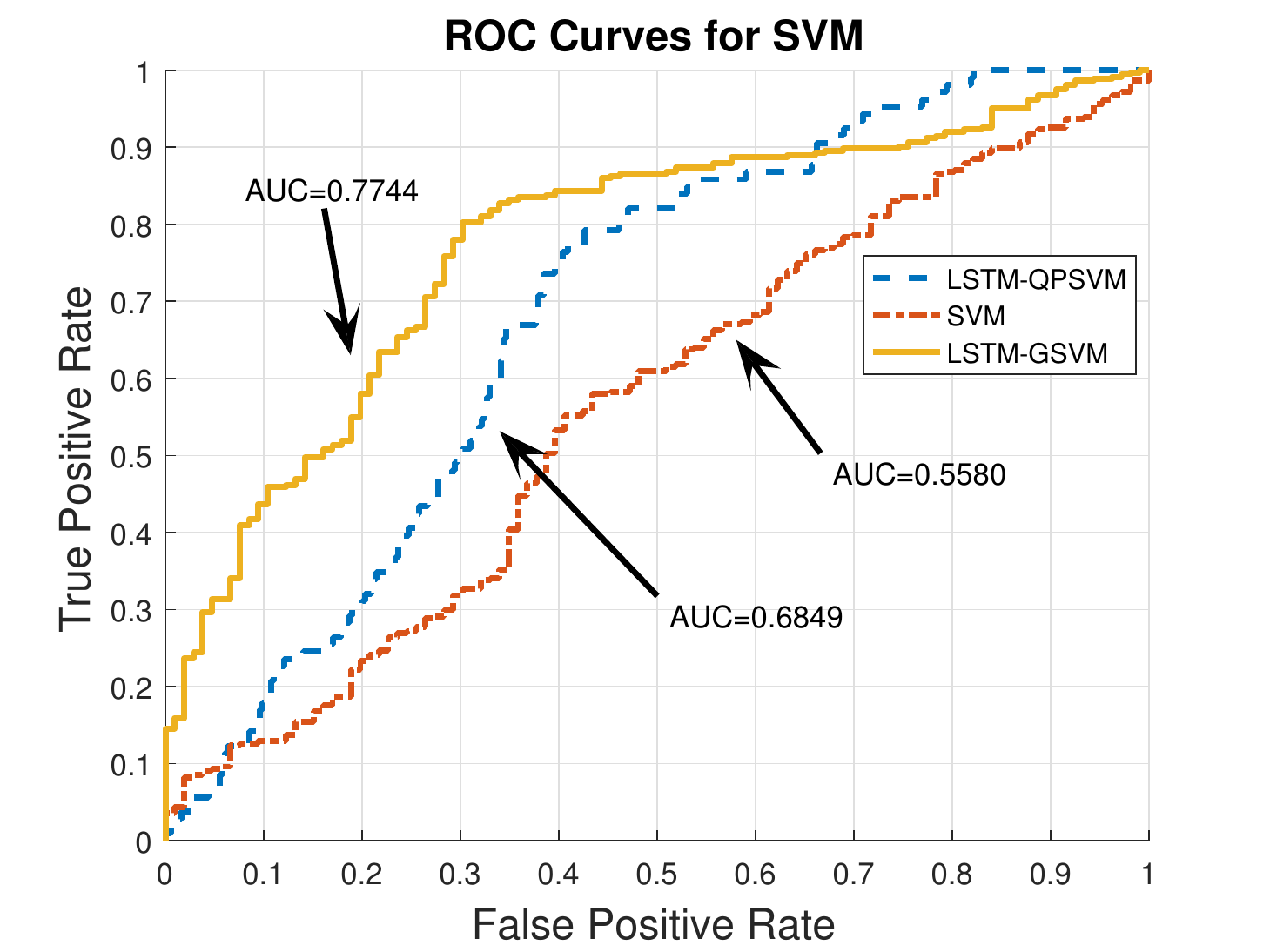}
		\caption{\centering\medskip} \label{svm1-7}
	\end{subfigure} \hspace*{\fill}
			\begin{subfigure}[t]{0.45\textwidth}
		\centering
		\includegraphics[width=1.11\textwidth, height=0.78\textwidth]{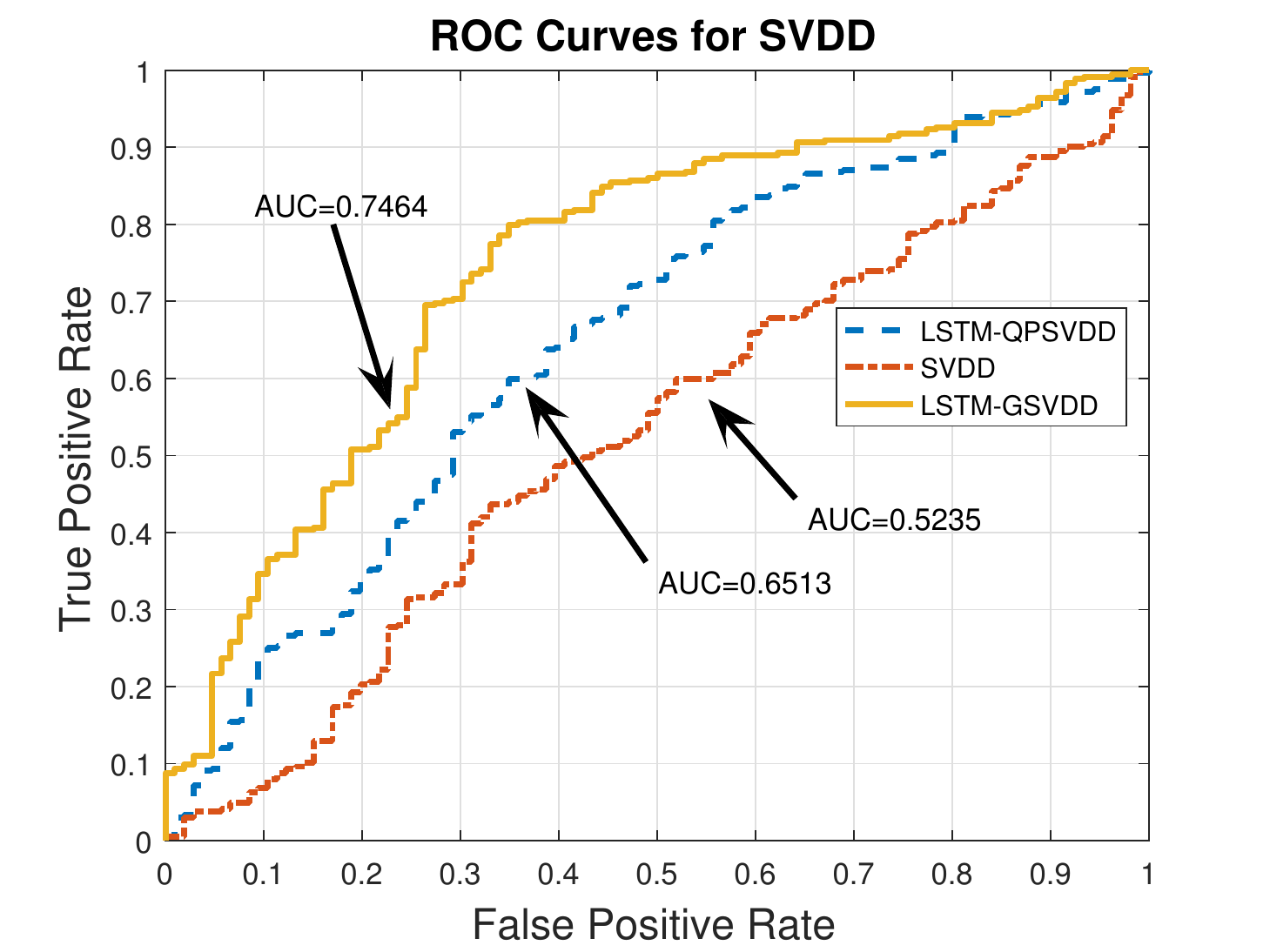}
		\caption{\centering\medskip} \label{svdd1-7}
	\end{subfigure} \hspace*{\fill}
	\medskip
	\caption{The ROC curves of the algorithms for the digit dataset, where we consider digit ``1" as normal and digit ``7" as anomaly (a) for the SVM based algorithms and (b) for the SVDD based algorithms.}\label{figs2}
	\end{figure*}

Here, we use area under ROC curve as a performance metric \cite{auc}. In a ROC curve, we plot true positive rate (TPR) as a function of false positive rate (FPR). Area under this curve, i.e., also known as AUC, is a well known performance measure for anomaly detection tasks \cite{auc}. In Fig. \ref{svm0-9} and \ref{svdd0-9}, we illustrate the ROC curves and provide the corresponding AUC scores, where we label digit ``0" and ``9" as normal and anomaly, respectively. For the OC-SVM and SVDD algorithms, since we directly take the mean of variable length data sequences to obtain fixed length sequences, they achieve significantly lower AUC scores compared to the introduced LSTM based methods. Among the LSTM based methods, LSTM-GSVM slightly outperforms LSTM-QPSVM. On the other hand, LSTM-GSVDD achieves significantly higher AUC than LSTM-QPSVDD. Since the quadratic programming based training method depends on the separated consecutive updates of the LSTM and SVM (or SVDD) parameters, it might not converge to even a local minimum. However, the gradient based method can guarantee convergence to at least a local minimum given a proper choice of the learning rate \cite{sayed}. Thus, although these methods might provide similar performances as in Fig. \ref{svm0-9}, it is also expected to obtain much higher performance from the gradient based method for certain cases as in Fig. \ref{svdd0-9}. However, overall, the introduced algorithms provide significantly higher AUC than the conventional methods.

Besides the previous scenario, we also consider a scenario, where we label digit ``1" and ``7" as normal and anomaly, respectively. In Fig. \ref{svm1-7} and \ref{svdd1-7}, we illustrate the ROC curves and provide the corresponding AUC scores. As in the previous scenario, for both the SVM and SVDD cases, the introduced algorithms achieve higher AUC scores than the conventional algorithms. Among the introduced algorithms, LSTM-GSVM and LSTM-GSVDD achieve the highest AUC scores for the SVM and SVDD cases, respectively. Furthermore, the AUC score of each algorithm is much lower compared to the previous case due to the similarity between digit ``1" and ``7".
	
	\subsection{Benchmark Real Datasets}
In this section, we compare the AUC scores of each algorithm on several different real benchmark datasets. Since our approach in this paper is generic, in addition to the LSTM based algorithms, we also implement our approach on the recently introduced RNN architecture, i.e., the GRU architecture, which is defined by the following equations \cite{gru}:
\begin{align}
\label{gruz}
&\vec{\tilde{z}}_{i,j}=\sigma\left(\vec{W}^{(\tilde{z})}\vec{x}_{i,j}+\vec{R}^{(\tilde{z})}\vec{h}_{i,j-1} \right) \\ \label{grur}
&\vec{r}_{i,j}=\sigma\left(\vec{W}^{(r)}\vec{x}_{i,j}+\vec{R}^{(r)}\vec{h}_{i,j-1} \right)\\\label{gruo}
&\vec{\tilde{h}}_{i,j}=g\left(\vec{W}^{(\tilde{h})}\vec{x}_{i,j}+\vec{r}_{i,j} \odot (\vec{R}^{(\tilde{h})}\vec{h}_{i,j-1}) \right)\end{align}\begin{align}\label{gruy}
&\vec{h}_{i,j}=\vec{\tilde{h}}_{i,j} \odot \vec{\tilde{z}}_{i,j} + \vec{h}_{i,j-1} \odot (\vec{1}-\vec{\tilde{z}}_{i,j}), 
\end{align}  
where $\vec{h}_{i,j} \in \mathbb{R}^{m}$ is the output vector and $\vec{x}_{i,j} \in \mathbb{R}^{p}$ is the input vector. Furthermore, $\vec{W}^{(\cdot)}$ and $\vec{R}^{(\cdot)}$ are the parameters of the GRU, where the sizes are selected according to the dimensionality of the input and output vectors. We then replace \cref{eq:z,eq:s,eq:f,eq:state,eq:o,eq:output} with \cref{gruz,grur,gruo,gruy} in Fig. \ref{pooling} to obtain GRU based anomaly detectors.

We first evaluate the performances of the algorithms on the occupancy dataset \cite{occupancy}. In this dataset, we have five features, which are relative humidity percentage, light (in lux), carbon dioxide level (in ppm), temperature (in Celsius) and humidity ratio, and our aim is to determine whether an office room is occupied or not based on the features. Here, we use the same procedure with the previous subsection to separate the test and training data. Moreover, using the training data, we select $\mu=0.05, 0.05,0.001$ and $0.01$ for LSTM-GSVM, LSTM-QPSVM, LSTM-GSVDD and LSTM-QPSVDD, respectively. Note that, for the GRU based algorithms in this subsection, we use the same parameter setting with the LSTM based algorithms. Furthermore, we choose $m=5$ and $\lambda=0.5$ for all of the experiments in this subsection in order to maximize the performances of the algorithms.

As can be seen in Table \ref{tab:alldatasets}, due to their inherent memory, both the LSTM and GRU based algorithms achieve considerably high AUC scores compared to the conventional SVM and SVDD algorithms. Moreover, GRU-GSVDD achieves the highest AUC score among all the algorithms, where the LSTM based algorithms (LSTM-GSVM and LSTM-QPSVM) also provide comparable AUC scores. Here, we also observe that the gradient based training method provides higher AUC scores compared to the quadratic programming based training method, which might stem from its separated update procedure that does not guarantee convergence to a certain local minimum.

Other than the occupancy dataset, we also perform an experiment on the HKE rate dataset in order to examine the performances for a real life financial scenario. In this dataset, we have the amount of Hong Kong dollars that one can buy for one US dollar on each day. In order to introduce anomalies to this dataset, we artificially add samples from a Gaussian distribution with the mean and ten times the variance of the training data. Furthermore, using the training data, we select $\mu=0.01, 0.005,0.05$ and $0.05$ for LSTM-GSVM, LSTM-QPSVM, LSTM-GSVDD and LSTM-QPSVDD, respectively. 

In Table \ref{tab:alldatasets}, we illustrate the AUC scores of the algorithms on the HKE rate dataset. Since we have time series data, both the LSTM and GRU based algorithms naturally outperform the conventional methods thanks to their inherent memory, which preserves sequential information. Moreover, since the LSTM architecture also controls its memory content via output gate unlike the GRU architecture \cite{gru}, we obtain the highest AUC scores from LSTM-GSVM. As in the previous cases, the gradient based training method provides better performance than the quadratic programming based training.

We also evaluate the AUC scores of the algorithms on the http dataset \cite{http}. In this dataset, we have 4 features, which are duration (number of seconds of the connection), network service, number of bytes from source to destination and from destination to source. Using these features, we aim to distinguish normal connections from network attacks. In this experiment, we select $\mu=0.01, 0.05,0.001$ and $0.01$ for LSTM-GSVM, LSTM-QPSVM, LSTM-GSVDD and LSTM-QPSVDD, respectively. 

We demonstrate the performances of the algorithms on the http dataset in Table \ref{tab:alldatasets}. Even though all the algorithms achieve high AUC scores on this dataset, we still observe that the LSTM and GRU based algorithms have higher AUC scores than the conventional SVM and SVDD methods. Overall, GRU-QPSVDD achieves the highest AUC score and the quadratic programming based training methods performs better than the gradient based training method on this dataset. However, since the AUC scores are very high and close to each other, we observe only slight performance improvement for our algorithms in this case.

 \begin{table*}[h]
 \centering
 
   \caption{AUC scores of the algorithms for the occupancy, HKE rate, http and Alcoa stock price datasets.}
 \resizebox{2\columnwidth}{.15\columnwidth}{
  \begin{tabular}{|l||*{10}{c|}} \hline
  \backslashbox[3em]{Datasets \kern-2em}{\kern-2em Algorithms}
  &\makebox{SVM}&\makebox{SVDD}&\makebox{LSTM-QPSVM}&\makebox{LSTM-GSVM}&\makebox{LSTM-QPSVDD}&\makebox{LSTM-GSVDD}&\makebox{GRU-QPSVM}&\makebox{GRU-GSVM}&\makebox{GRU-QPSVDD}&\makebox{GRU-GSVDD} \\\hline\hline
 \hspace{.15cm} Occupancy  & $0.8676$& $0.6715$ & $0.8917$& $0.8957$ & $0.7869$& $0.8609$&$ 0.8718$&$0.9049 $&$ 0.7217$&$0.9099 $ \\\hline
  \hspace{.25cm} HKE  & $0.8000$& $0.8500$ & $0.9467$& $0.9783$ & $0.8560$& $0.9753$& $ 0.8479$&$0.9516 $&$ 0.8791$&$0.9517 $ \\\hline
  \hspace{.25cm} Http    & $ 0.9963$& $0.9993 $ & $0.9992 $& $0.9983 $ & $0.9994 $& $0.9994 $&$ 0.9986 $&$ 0.9989 $&$ 0.9999$&$ 0.9994 $ \\\hline
  \hspace{.15cm} Alcoa     & $0.7197 $& $0.9390 $ & $0.9496 $& $ 0.9515$ & $ 0.9415$& $0.9507 $&$ 0.7581 $&$ 0.9392 $&$ 0.9651$&$  0.9392$  \\\hline
   \end{tabular}
  }
\label{tab:alldatasets}
\end{table*}

As the last experiment, we evaluate the anomaly detection performances of the algorithms on another financial dataset, i.e., the Alcoa stock price dataset \cite{alcoa}. In this dataset, we have daily stock price values. As in the HKE rate dataset, we again artificially introduce anomalies via a Gaussian distribution with the mean and ten times the variance of the training data. Moreover, we choose $\mu=0.01, 0.001,0.001$ and $0.005$ for LSTM-GSVM, LSTM-QPSVM, LSTM-GSVDD and LSTM-QPSVDD, respectively. 

In Table \ref{tab:alldatasets}, we illustrate the AUC scores of the algorithms on the Alcoa stock price dataset. Here, we observe that the GRU and LSTM based algorithms achieve considerably higher AUC scores than the conventional methods thanks to their memory structure. Although the LSTM based algorithms have higher AUC in general, we obtain the highest AUC score from GRU-QPSVDD. Moreover, as in the previous experiments, the gradient based training method provides higher performance compared to the quadratic programming based method thanks to its learning capabilities.

\section{Concluding Remarks}\label{sec:conclusion}
In this paper, we study anomaly detection in an unsupervised framework and introduce LSTM based algorithms. Particularly, we have introduced a generic LSTM based structure in order to process variable length data sequences. After obtaining fixed length sequences via our LSTM based structure, we introduce a scoring function for our anomaly detectors based on the OC-SVM \cite{svm1} and SVDD \cite{svdd} algorithms. As the first time in the literature, we jointly optimize the parameters of both the LSTM architecture and the final scoring function of the OC-SVM (or SVDD) formulation. To jointly optimize the parameters of our algorithms, we have also introduced gradient and quadratic programming based training methods with different algorithmic merits, where we extend our derivations for these algorithms to the semi-supervised and fully supervised frameworks. In order to apply the gradient based training method, we modify the OC-SVM and SVDD formulations and then provide the convergence results of the modified formulations to the actual ones. Therefore, we obtain highly effective anomaly detection algorithms, especially for time series data, that are able to process variable length data sequences. In our simulations, due to the generic structure of our approach, we have also introduced GRU based anomaly detection algorithms. Through extensive set of experiments, we illustrate significant performance improvements achieved by our algorithms with respect to the conventional methods \cite{svm1,svdd} over several different real and simulated datasets.
\begin{spacing}{.87}
\bibliographystyle{IEEEtran}
\bibliography{my_ref}
\end{spacing}
\end{document}